%% file: main.tex
\documentclass[acmsmall,screen,nonacm]{acmart}

\usepackage{amsmath,amsfonts}
\usepackage[ruled,vlined]{algorithm2e}
\usepackage{adjustbox}
\usepackage{graphicx}
\usepackage{textcomp}
\usepackage{xcolor}
\usepackage{physics}
\usepackage{wrapfig}
\usepackage{subcaption}
\usepackage{diagbox}
\usepackage{multirow}
\usepackage{multicol}
\usepackage{makecell}
\usepackage{enumitem}
\usepackage{hyperref}
\usepackage{bm}
\hypersetup{}

\newtheorem{theorem}{Theorem}[section]

\newtheorem{lemma}{Lemma}[section]
\newtheorem{remark}{Remark}[section]
\newtheorem{definition}{Definition}[section]
\newtheorem{proposition}{Proposition}[section]

\newcommand{\myCompilerName}{GeneCS}
\newcommand{\myCompilerNameSpace}{GeneCS }

\begin{document}
\title{\myCompilerName: Synthesizing Resource-Efficient Code Surgery for Arbitrary Quantum Stabilizer Codes}

\author{Junyu Zhou}
\affiliation{%
  \institution{University of Pennsylvania}
  \city{Philadelphia}
  \country{USA}}
\email{junyuzh@seas.upenn.edu}

\author{Ali Javadi-Abhari}
\affiliation{%
  \institution{IBM}
  \city{Yorktown Heights}
  \country{USA}}
\email{ali.javadi@ibm.com}

\author{Gushu Li}
\affiliation{%
  \institution{University of Pennsylvania}
  \city{Philadelphia}
  \country{USA}}
\email{gushuli@seas.upenn.edu}

\begin{abstract}
Efficiently realizing logical operations on general stabilizer codes remains a long-standing challenge in fault-tolerant quantum computing. While code surgery provides a general framework with provable guarantees through joint logical measurements, existing constructions are largely theoretical and incur substantial ancilla overhead in practice. We present \myCompilerName, a resource-efficient synthesizer of code surgery protocols for arbitrary stabilizer codes, designed for the backend of fault-tolerant quantum compilers. Our approach leverages structure-aware optimizations to eliminate redundancy in graph construction, dynamically balance expansion and congestion, and incorporate code degree constraints. Experimental results show that \myCompilerNameSpace reduces ancillary qubits and checks by over 85\% on average for both single-code and cross-code logical operations, while the optimized constructions exhibit logical-error-rate behavior comparable to the baseline under BP-OSD (belief propagation with ordered-statistics decoding). On surface-code stress tests with more than $10^4$ physical qubits, one independent synthesis run takes about one second on average. These results enable practical logical operations and efficient cross-code communication for modern QLDPC (quantum low-density parity-check) codes and heterogeneous code systems.
\end{abstract}




\maketitle 

\input{tex/1-Introduction}
\input{tex/2-Preliminary}
\input{tex/3-ProblemFormulation}
\input{tex/4-OptimizingExpanderConstruction}
\input{tex/5-BalancingExpansionandCongestion}
\input{tex/6-IncorporatingDegreeConstraints}
\input{tex/7-Evaluation}
\input{tex/8-RelatedWork}
\input{tex/9-Statements}


\bibliographystyle{ACM-Reference-Format}
\bibliography{refs}

\newpage
\appendix

\input{tex/Appendix}

\end{document}

%% file: tex/1-Introduction.tex
\section{Introduction}\label{Sec:Introduction}

Quantum error correction (QEC)~\cite{gottesman1997stabilizer} is a fundamental requirement for scalable fault-tolerant~\cite{gottesman1998theory} quantum computing. Most practical QEC schemes are based on the stabilizer formalism, which provides a unified framework for describing a wide range of codes. Over the past decade, a diverse landscape of stabilizer codes has emerged, including topological codes~\cite{bombin2013introduction} such as the surface code~\cite{fowler2012surface} and color codes~\cite{landahl2011fault}, as well as more recent constructions such as Quantum Low-Density Parity-Check (QLDPC) codes~\cite{breuckmann2021quantum}, e.g., Bivariate Bicycle (BB) codes~\cite{bravyi2024high,yoder2025tour}. These code families offer complementary advantages: surface codes enable efficient fault-tolerant logical operations with local interactions, while QLDPC codes provide significantly improved encoding rates and are attractive for quantum memory.

A key implication of this diversity is that practical fault-tolerant quantum computing is unlikely to rely on a single code family. Due to fundamental constraints on fault-tolerant universality~\cite{eastin2009restrictions,chakraborty2026nogo}, efficient implementations of universal fault-tolerant quantum computing typically require combining multiple codes or switching between encoding schemes. Emerging system-level studies further support this view, showing that leveraging different codes or platforms can significantly improve resource efficiency~\cite{stein2025hetec,fang2026bridging,mundada2026heterogeneous}.
These trends highlight a fundamental requirement: \emph{efficient cross-code communication}. Logical quantum information must be transferred across different encodings, and joint logical operations must be performed between them. Consequently, enabling \emph{communication and logical operations across arbitrary stabilizer codes} becomes a central challenge for scalable fault-tolerant quantum computing.

This paper focuses on this challenge. In particular, we study how to realize fault-tolerant logical operations that enable cross-code communication, and more specifically, how to efficiently synthesize \emph{code surgery} protocols for arbitrary stabilizer codes. Code surgery~\cite{cohen2022low,cowtan2024ssip,cross2024improved}, a generalization of lattice surgery~\cite{horsman2012surface}, provides a natural mechanism for implementing joint logical measurements and transferring logical information across arbitrary stabilizer codes.
While lattice surgery on the surface code is well understood and highly optimized~\cite{horsman2012surface,fowler2018low,zhou2026topols,litinski2019game}, efficiently extending these techniques to general stabilizer codes remains an open challenge.
Existing approaches are either general but inefficient---relying on conservative constructions that incur substantial ancilla overhead, with the first theoretical guarantees of QLDPC structure and distance preservation given in~\cite{williamson2024low}---or specialized ad-hoc designs~\cite{williamson2024low} that are limited to specific instances, as well as constructions that rely on expensive computations such as code distance checking~\cite{ide2025fault,zheng2025high,cowtan2024ssip}, and thus do not scale.

This gap highlights a missing abstraction: \emph{resource-efficient code surgery should be treated as a synthesis and optimization problem}. In particular, we observe that constructing code surgery protocols can be naturally formulated as a \emph{graph synthesis problem}, where the ancilla system corresponds to a graph whose structure determines both correctness (e.g., distance preservation) and resource cost. Existing constructions~\cite{williamson2024low,he2025extractors} follow rigid pipelines (expansion, decongestion, thickening, and cellulation) designed for analysis, but introduce significant redundancy when applied in practice. This motivates the need for a compiler that directly optimizes this graph structure.

\begin{figure}[t]
    \centering
    \includegraphics[width=1.0\linewidth]{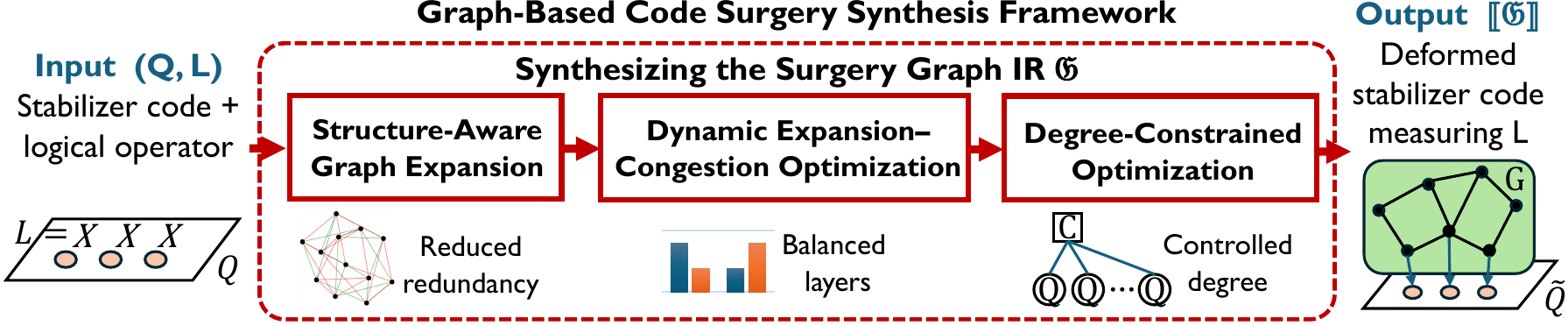}
    \vspace{-10pt}
    \caption{Overview of \myCompilerName~workflow}
    \label{fig:overview}
    \vspace{-6pt}
\end{figure}

To this end, we propose \myCompilerName, a synthesizer framework that can generate resource-efficient code surgery protocols for arbitrary stabilizer codes and serves as a backend component of fault-tolerant quantum compilers (Section~\ref{Sec:GeneralCodeSurgery}). The central abstraction in \myCompilerNameSpace is a \emph{surgery graph IR}, which exposes the graph construction underlying code surgery to compiler optimization. \myCompilerNameSpace optimizes the generated graph while preserving code distance and QLDPC structure by construction.
Our design is based on three coupled ideas. \textbf{First}, conditioned graph expansion reuses connectivity already present in the path-matching graph, avoiding the redundant edges introduced by an independent expander. \textbf{Second}, dynamic cycle tracking monitors congestion as the graph is constructed and identifies an expansion--congestion balance point, avoiding unnecessary thickening layers. \textbf{Third}, degree-aware cycle partitioning and cellulation use the available bounded-degree budget to reduce ancillary qubits and checks. These techniques operate on shared synthesis state rather than as independent optimization passes.

Our experimental results show that \myCompilerNameSpace substantially reduces the resource overhead of code surgery. For single-code logical measurements, it reduces ancillary qubits by $86.7\%$ and checks by $85.8\%$ on average relative to the best baseline. For cross-code logical measurements, the reductions are approximately $85\%$. The optimized constructions exhibit logical-error-rate behavior comparable to the baseline under the same BP-OSD (belief propagation with ordered-statistics decoding) decoder~\cite{roffe2020decoding} and noise model, while the part of the syndrome-extraction circuits added by surgery needs about 85\% fewer two-qubit gates and fault locations (about half for the whole circuit, which also measures the unchanged checks of the original codes). On the surface-code stress test with more than $10^4$ physical qubits, one independent synthesis run takes approximately one second on average.

The major contributions of \myCompilerNameSpace can be summarized as follows:
\begin{enumerate}

\item \textbf{A surgery graph IR for code surgery synthesis.}
We expose graph-based code surgery construction as an intermediate representation with explicit correctness conditions and resource objectives.

\item \textbf{Coupled graph synthesis and optimization.}
We introduce conditioned expansion, dynamic cycle tracking, and degree-aware partitioning and cellulation to jointly reduce edges, layers, and checks.

\item \textbf{An implementation and evaluation of \myCompilerName.}
\myCompilerNameSpace reduces ancillary resources by over $85\%$ on average across single-code and cross-code logical measurements while retaining by-construction guarantees and comparable observed logical-error-rate behavior.

\end{enumerate}

%% file: tex/2-Preliminary.tex
\section{Preliminary} \label{Sec:Preliminary}

This section reviews the background needed for the paper: Pauli measurements, stabilizer codes, graph terminology, and code surgery. We focus on the algebraic objects manipulated by \myCompilerName; for a broader introduction to quantum computing, we recommend~\cite{nielsen2010quantum}, and for stabilizer codes,~\cite{gottesman1997stabilizer,bradshaw2026introduction}.

\subsection{Qubits and Pauli Measurements}

A \emph{qubit} is the basic unit of quantum information. Its state space is the two-dimensional complex Hilbert space $\mathbb{C}^2$, with computational basis $\ket{0}$ and $\ket{1}$. A pure qubit state has the form $\alpha\ket{0}+\beta\ket{1}$, where $|\alpha|^2+|\beta|^2=1$, and an $n$-qubit state lies in $(\mathbb{C}^2)^{\otimes n}$. We introduce this state-space notation only to fix terminology: the constructions in this paper are stated in terms of Pauli operators, their supports, and their commutation relations.

The four single-qubit Pauli operators are
\[
I=\begin{bmatrix}1&0\\0&1\end{bmatrix},\quad
X=\begin{bmatrix}0&1\\1&0\end{bmatrix},\quad
Z=\begin{bmatrix}1&0\\0&-1\end{bmatrix},\quad
Y=iXZ.
\]
Tensor products of these operators describe the multi-qubit observables used throughout stabilizer error correction.

\begin{definition}[Pauli string, support, and weight]
An $n$-qubit \emph{Pauli string} is an operator
\[
P=\alpha P_1\otimes \cdots \otimes P_n,
\]
where $\alpha\in\{\pm 1,\pm i\}$ and each $P_i\in\{I,X,Y,Z\}$. Its support is
$\operatorname{supp}(P)=\{i:P_i\neq I\}$, and its weight is
$\operatorname{wt}(P)=|\operatorname{supp}(P)|$.
\end{definition}

The stabilizer checks $S_i$ and logical operators $L$ below are Pauli strings in this sense. Two Pauli strings either commute or anti-commute. A Hermitian Pauli string, namely one with phase $\pm 1$, has eigenvalues $\pm 1$ and can be measured as an observable. Measuring such a string returns one classical bit of information, conventionally written as an eigenvalue in $\{+1,-1\}$, and projects the quantum state to the corresponding eigenspace. 

\subsection{Stabilizer Codes and Logical Operators} \label{Sec:Stabilizer Formalism}
Quantum error correction (QEC) protects quantum information by encoding a small number of \emph{logical} qubits into a larger number $n$ of noisy \emph{physical} qubits. The extra physical qubits provide redundancy: many errors can be detected by measuring auxiliary constraints without directly measuring the encoded logical state itself. A stabilizer code is an algebraic way to specify such an encoding using commuting Pauli measurements.

Let $\mathcal{P}_n$ denote the $n$-qubit Pauli group, i.e., all $n$-qubit Pauli strings with phases in $\{\pm 1,\pm i\}$. A stabilizer code $\mathcal{Q}$~\cite{gottesman1997stabilizer} is defined by an abelian subgroup generated by commuting Pauli strings $S_1,\ldots,S_m$:
\begin{equation*}
    \mathcal{S} = \langle S_1, \dots, S_m \rangle \subset \mathcal{P}_n, \quad -I \notin \mathcal{S}.
\end{equation*}
The code space is the simultaneous $+1$ eigenspace of all Pauli strings in $\mathcal{S}$ (the stabilizers), given by:
\begin{equation*}
    \mathsf{Code}(\mathcal{Q}) = \{\, |\psi\rangle : S |\psi\rangle = |\psi\rangle,\ \forall S \in \mathcal{S} \,\}.
\end{equation*}
States in $\mathsf{Code}(\mathcal{Q})$ are called code states. Intuitively, a stabilizer is a constraint that every valid encoded state satisfies. Stabilizer error correction uses repeated Pauli measurements as checks on the encoded state. When error occurs, some measured stabilizer outcomes may flip; the resulting vector of outcomes is the information used by a decoder to infer corresponding correction operations.

\begin{definition}[Check, syndrome, code parameters, and distance]
The generators $S_1,\ldots,S_m$ are stabilizer \emph{checks}. The \emph{syndrome} of a state is the vector of outcomes obtained by measuring these checks. If the code acts on $n$ physical qubits and the stabilizer group has rank $r$, then it encodes $k=n-r$ logical qubits. The distance of a stabilizer code is
\[
d=\min\{\operatorname{wt}(P):P\in\mathcal{N}(\mathcal{S}) / \mathcal{S}\},
\]
where $\mathcal{N}(\mathcal{S})$${}=\{P\in\mathcal{P}_n: PS=SP,\ \forall S\in\mathcal{S}\}$ is the  normalizer of $\mathcal{S}$ in $\mathcal{P}_n$. Thus $d$ is the minimum weight of a nontrivial logical Pauli operator (an element of $\mathcal{N}(\mathcal{S})/\mathcal{S}$).
We use the standard notation $\llbracket n,k,d\rrbracket$ for a quantum code that uses $n$ physical qubits to encode $k$ logical qubits with distance $d$.
\end{definition}

Logical Pauli operators are the Pauli operations visible at the encoded level. They preserve the code space but are not themselves stabilizers. Specifically, a Pauli operator $L \in \mathcal{P}_n$ is logical if it commutes with all stabilizers:
$LS = SL$, $\forall S \in \mathcal{S}$,
and $L \notin \mathcal{S}$.  Then the logical Pauli operators correspond to equivalence classes in the quotient space
$\mathcal{N}(\mathcal{S}) / \mathcal{S}$.
Later, a logical measurement will be specified by one representative Pauli string $L$ from such an equivalence class. Its support identifies the physical qubits involved in the measurement, and its local Pauli factors determine how stabilizer checks interact with that measurement.


    
\begin{wrapfigure}{R}{0.26\textwidth}
    \centering
    \vspace{-5pt}
    \includegraphics[width=0.26\textwidth]{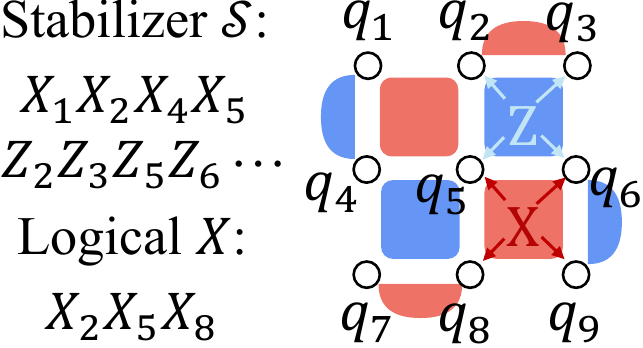}
    \vspace{-15pt}
    \caption{$\llbracket 9,1,3\rrbracket$ Surface code}
    \label{fig:surface}
    \vspace{-5pt}
\end{wrapfigure}
\textbf{Example:} Figure~\ref{fig:surface} shows a distance-three surface code with nine data qubits. Each colored plaquette defines a stabilizer generator acting on the four adjacent qubits. For example, the upper-left red plaquette corresponds to the $X$-type stabilizer $X_1X_2X_4X_5$, while the upper-right blue plaquette corresponds to the $Z$-type stabilizer $Z_2Z_3Z_5Z_6$. There are also weight-2 stabilizers at the boundary, such as $Z_1Z_4$ (blue, left) and $X_2X_3$ (red, top). A logical Pauli operator commutes with every stabilizer but is not itself a stabilizer. As illustrated, $X_2X_5X_8$ is one representative logical $X$ operator: it overlaps every blue ($Z$-type) plaquette on an even number of qubits.


\begin{definition}[CSS and QLDPC codes]
A stabilizer code is a \emph{Calderbank--Shor--Steane (CSS) code} when its checks can be generated by $X$-only checks and $Z$-only checks. It is a \emph{quantum low-density parity-check} (QLDPC) code~\cite{breuckmann2021quantum} when both check weight and qubit degree are bounded by constants independent of the block length: each check acts on $O(1)$ qubits, and each physical qubit participates in $O(1)$ checks.
\end{definition}

The surface-code example in Figure~\ref{fig:surface} is both CSS and QLDPC: its plaquette checks are purely $X$- or $Z$-type, every check has bounded weight, and every physical qubit has bounded degree independent of the lattice size.

\subsection{Graph-Theoretic Preliminaries}\label{Sec:GraphConcepts}

We use a small amount of graph terminology throughout the compiler construction. All graphs in the paper are finite, undirected, and loopless; parallel edges are allowed. For a graph $G=(V,E)$, $V$ denotes the vertex set and $E$ denotes the multiset of edge instances. Thus parallel edge instances are counted separately in $|E|$. Given $U\subseteq V$, the \emph{edge boundary} $\delta(U)$ is the multiset of edge instances with exactly one endpoint in $U$. Intuitively, choosing $U$ cuts the graph into two parts, and $|\delta(U)|$ counts the edges that cross the cut, including multiplicity. A small boundary means that the subgraph induced by $U$ is connected to the rest of the graph through a narrow bottleneck.

\begin{definition}[Spanning forest and fundamental cycle]\label{def:spanning-forest}
A \emph{forest} is a graph with no cycles. A \emph{spanning forest} of $G=(V,E)$ is a subgraph $T=(V,E_T)$ that is a tree on each connected component of $G$; if $G$ is connected, $T$ is a spanning tree. For an edge $e=(u,v)\in E\setminus E_T$ whose endpoints lie in the same component of $T$, the edge $e$ together with the unique path from $u$ to $v$ in $T$ forms the \emph{fundamental cycle} of $e$ with respect to $T$.
\end{definition}

A spanning forest captures the part of the graph that connects vertices without creating any cycles. Each edge outside the forest closes exactly one loop with the forest path between its endpoints; this loop is the edge's fundamental cycle. Thus the non-forest edges identify the basic cycles introduced on top of the forest.

\begin{definition}[Relative expansion]\label{def:relative-expansion}
The Cheeger constant $\beta(G)$ is the largest value such that
\[
|\delta(U)|\ge \beta(G)\min(|U|,|V\setminus U|)
\]
for every $U\subseteq V$. Given a distinguished vertex set $P\subseteq V$ and a scale parameter $t>0$, the relative Cheeger constant $\beta_t(G,P)$ is the largest value such that
\[
|\delta(U)|\ge \beta_t(G,P)\min(t, |U\cap P|, |P\setminus U|)
\]
for every $U\subseteq V$.
\end{definition}

Expansion measures how well connected the graph is across all possible cuts. A graph with large $\beta(G)$ has no sparse cut: every reasonably sized set of vertices has many edges leaving it. The relative version $\beta_t(G,P)$ applies the same idea to a distinguished subset $P$. Rather than requiring every vertex set to expand according to its total size, it asks that any cut separating vertices of $P$ from other vertices of $P$ must expose enough crossing edges, up to the scale $t$. Thus relative expansion rules out bottlenecks among the distinguished vertices.

\begin{definition}[Cycle space and cycle basis]\label{def:cycle-basis}
For an undirected graph $G=(V,E)$, the cycle space of $G$ is the vector space over $\mathbb{F}_2$ generated by cycles, where addition is symmetric difference of edge instances. In particular, two parallel edges form a cycle of length two. A \emph{cycle basis} is any basis of this vector space. Equivalently, each cycle in the graph can be represented as an $\mathbb{F}_2$-sum of basis cycles.
\end{definition}

The cycle space records the closed loops in the graph. Addition by symmetric difference means that an edge appearing in two loops cancels, so combining loops keeps only the edges used an odd number of times. A cycle basis is therefore a compact set of fundamental loops from which all other loops can be assembled. Different bases can describe the same graph, but they may have very different sizes and overlaps.

\begin{definition}[Cycle representation]\label{def:cycle-representation}
A \emph{cycle representation} of a graph is a finite set of cycles together with the $\mathbb{F}_2$-subspace that they span. It is \emph{complete} when this span is the full cycle space of the graph. Unlike a cycle basis, a cycle representation may contain redundant cycles; this flexibility is useful when replacing one generating set of cycles by another that spans the same subspace.
\end{definition}

\begin{definition}[Cycle-basis congestion]\label{def:cycle-basis-congestion}
Let $\mathcal{R}$ be a cycle basis of $G$. For an edge $e\in E$, let $\rho_{\mathcal{R}}(e)=|\{C\in\mathcal{R}:e\in C\}|$. The congestion of $\mathcal{R}$ is $\rho(\mathcal{R})=\max_{e\in E}\rho_{\mathcal{R}}(e).$
\end{definition}

Congestion measures how much the chosen fundamental loops overlap. If many basis cycles all pass through the same edge, that edge is a hotspot of the cycle basis. Low congestion means that the cycle basis spreads its loops across the graph instead of repeatedly reusing a small number of edges. This distinction matters because two cycle bases can generate the same cycle space while placing very different loads on individual edges.


    

\begin{wrapfigure}[10]{r}{0.55\textwidth}
    \centering
\vspace{-8pt}
    \includegraphics[width=\linewidth]{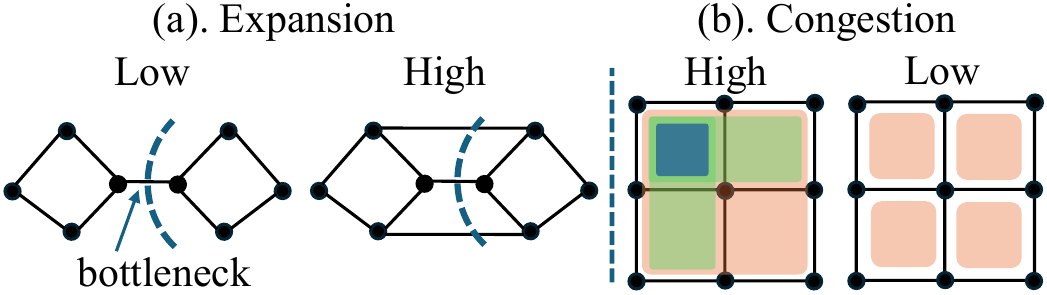}
    \vspace{-15pt}
    \caption{Expansion and congestion}
    \label{fig:Expansionandcongestion}
    \vspace{-40pt}
\end{wrapfigure}

\textbf{Example:} Figure~\ref{fig:Expansionandcongestion} illustrates the graph terms mentioned above. Figure~\ref{fig:Expansionandcongestion}(a) compares graphs with low and high expansion. In the low-expansion graph, the dashed cut passes through only a few edges, revealing a bottleneck that weakly connects the two sides and a smaller (relative) Cheeger constant. In contrast, every comparable cut in the high-expansion graph crosses many edges, and the (relative) Cheeger constant is larger. Figure~\ref{fig:Expansionandcongestion}(b) illustrates cycle-basis congestion. The two sets of colored blocks represent two different cycle bases for the same two graphs. Although both cycle bases span the same cycle space, the left basis repeatedly uses a central edge (the edges linked to the top left vertex), resulting in high congestion, whereas the right basis distributes cycles more evenly across the graph, giving low congestion.


\subsection{Codes as Graphs: A Unifying Perspective} \label{Sec:CodesasGraphs:AUnifyingView}

We now connect the graph terminology above to the ancilla systems used in code surgery. Each synthesized ancilla system is CSS and admits a graph incidence representation. This representation records which checks act on which ancilla qubits; it is algebraic, not a geometric embedding.

For this CSS ancilla system, we use the following dictionary:
\begin{itemize}
    \item vertices correspond to one family of ancilla stabilizer checks,
    \item edges correspond to ancilla qubits,
    \item cycles correspond to the complementary family of ancilla stabilizer checks.
\end{itemize}
Under this correspondence, the action of ancilla stabilizers is encoded directly in the graph structure: each vertex check acts on all incident edges of a vertex, while each cycle check acts on the set of edges forming a cycle. One may choose the $X$-check family as vertices and the $Z$-check family as cycles, or exchange the two roles. This representation translates ancilla-system structure into a combinatorial object that can be synthesized and optimized.

\begin{wrapfigure}{r}{0.45\textwidth}
    \centering
    \vspace{-5pt}
    \includegraphics[pagebox=cropbox,width=\linewidth]{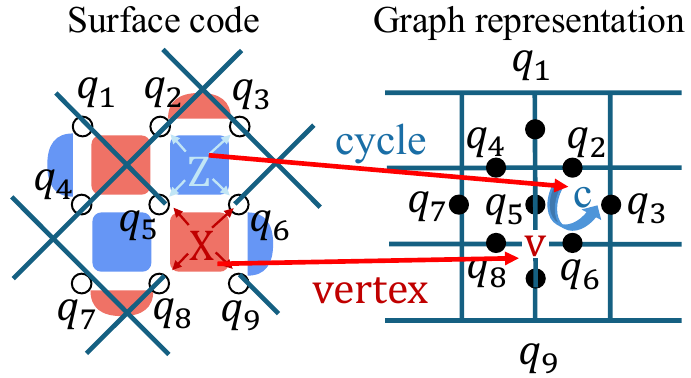}
    \vspace{-15pt}
    \caption{Example: surface code and its graph representation}
    \label{fig:asgraphsurface}
    \vspace{-10pt}
\end{wrapfigure}
We illustrate this abstraction using the surface code shown in Figure~\ref{fig:asgraphsurface}. On the left is a surface code patch, where red regions denote $X$-type stabilizers and blue regions denote $Z$-type stabilizers, each acting on a set of neighboring qubits (shown as hollow dots).
This structure naturally admits the graph interpretation shown on the right. In this graph, a red stabilizer becomes a vertex, a data qubit between two red stabilizers becomes an edge, and each blue stabilizer becomes the boundary cycle of a face. Thus Figure~\ref{fig:asgraphsurface} should be read as a dictionary between code objects and graph objects: local check incidence in the code becomes edge incidence in the graph. Under this mapping, each $X$ stabilizer acts on the edges incident to a vertex, while each $Z$ stabilizer acts on the edges along the boundary of a face.
The key point is algebraic: graph incidence determines check support. This perspective is the basis for the surgery graph IR in Section~\ref{Sec:SurgeryGraphIR}, where optimizing the graph directly reduces the qubits and checks in the synthesized ancilla system.

\subsection{General Code Surgery} \label{Sec:GeneralCodeSurgery}

Code surgery~\cite{cohen2022low, cross2024improved, he2025extractors, williamson2024low} realizes logical operations by measuring logical Pauli operators through an ancilla system. It does not apply a gate directly: it performs a nondestructive (joint) logical Pauli measurement, and logical gates are composed from such measurements.

In a fault-tolerant quantum compiler, a logical program is therefore decomposed into single-logical-qubit operations and joint logical Pauli measurements, from which gates such as CNOT are composed (see the example at the end of this section). For each logical Pauli measurement, \myCompilerNameSpace takes the stabilizer codes and the target logical operator and synthesizes a deformed code that performs the measurement through syndrome extraction; a code-to-circuit compiler then generates the physical syndrome-extraction circuits. \myCompilerNameSpace is thus a backend component of a fault-tolerant quantum compiler rather than a compiler for a high-level quantum language. It can synthesize only the measurements required by a given program or precompute reusable constructions for selected codes and operators.

Given a stabilizer code $\mathcal{Q}$ and a logical Pauli operator $L$, a surgery construction attaches an ancilla system to $\operatorname{supp}(L)$, producing a deformed code in which $L$ is a stabilizer. Its eigenvalue can then be obtained through syndrome extraction. Figure~\ref{fig:codesurgery} illustrates this deformation and its use in a joint measurement across two codes.

\begin{definition}[Code surgery instance]
A code surgery instance is a pair $(\mathcal{Q},L)$, where $\mathcal{Q}$ is a stabilizer code with stabilizer group $\mathcal{S}$ and $L$ is a logical Pauli operator of $\mathcal{Q}$. A surgery construction produces a deformed stabilizer code $\widetilde{\mathcal{Q}}$ by adding ancilla qubits and checks such that $L$ is included in the stabilizer group of $\widetilde{\mathcal{Q}}$ and can therefore be measured by syndrome extraction.
\end{definition}

Under the graph representation from the preceding subsection, the ancilla system is an ancilla graph attached to $\operatorname{supp}(L)$. Selected vertices, called \emph{ports}, connect the graph to the physical qubits on which $L$ acts. Graph edges represent ancilla qubits, while vertex and cycle checks provide the new stabilizer constraints. The vertex checks couple the ancilla graph to $L$, and affected checks of the original code are modified so that all checks of the deformed code commute. The middle panel of Figure~\ref{fig:codesurgery} illustrates this construction for Code Block 1; the same construction is applied independently to Code Block 2.

\begin{figure}[t]
    \centering
    \includegraphics[width=0.9\linewidth]{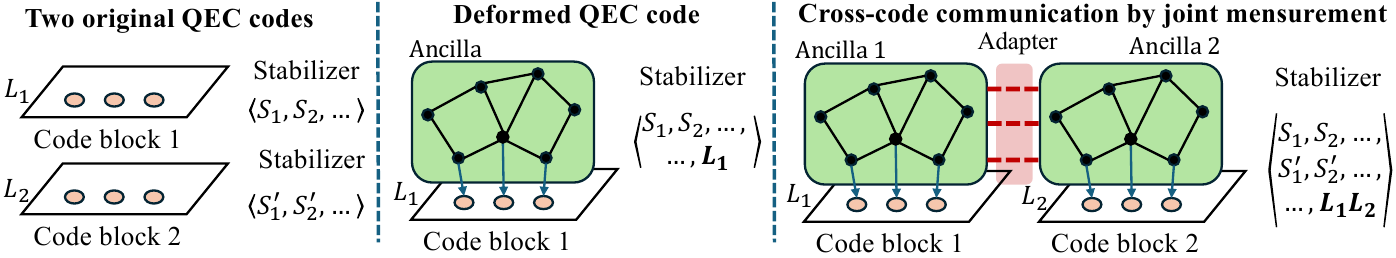}
    \vspace{-10pt}
    \caption{Overview of code surgery construction}
    \label{fig:codesurgery}
    \vspace{-6pt}
\end{figure}

To measure a joint logical operator $L_1L_2$ between codes $\mathcal{Q}_1$ and $\mathcal{Q}_2$, we first construct an ancilla graph for each pair $(\mathcal{Q}_i,L_i)$. We then connect the corresponding ports using an \emph{adapter}~\cite{swaroop2024universal}, or a \emph{bridge} when the codes belong to the same family. The resulting joint deformed code contains $L_1L_2$ as a stabilizer, whose eigenvalue is obtained through syndrome extraction. Such joint logical measurements serve as building blocks for state transfer and other communication protocols between codes. Their resource cost depends critically on the two ancilla graphs, motivating optimized ancilla-graph construction. For example, a logical CNOT from qubit $c$ of $\mathcal{Q}_1$ to qubit $t$ of $\mathcal{Q}_2$ uses an ancilla logical qubit $a$ prepared in $\ket{0}$: measuring the logical operators $X_aX_t$, then $Z_cZ_a$, and finally $X_a$, each followed by a Pauli correction ($Z_c$ or $X_t$) determined by its outcome, implements the CNOT~\cite{zhou2026topols}. \myCompilerNameSpace synthesizes the surgery for all three logical measurements, including the single-qubit measurement of $X_a$.

%% file: tex/3-ProblemFormulation.tex
\section{Surgery Graph IR and Synthesis Problem}
\label{Sec:SurgeryGraphIR}
\label{Sec:Motivation}

This section formulates the graph synthesis problem solved by \myCompilerName. We build on prior graph-based code-surgery constructions~\cite{williamson2024low,he2025extractors}, in which an ancilla graph attached to a logical operator denotes a deformed stabilizer code that measures that operator. \myCompilerNameSpace exposes this graph as an intermediate representation (IR) and synthesizes it subject to four objectives:
\begin{enumerate}
    \item \textbf{Well-formed deformed code.} The checks of the deformed code must mutually commute. The graph construction achieves this by connecting the local anti-commutations of each input check in pairs, using the path matchings defined below.
    \item \textbf{Distance preservation.} The deformed code must have distance at least the target distance of the input instance. Prior graph-based constructions reduce this requirement to a relative-expansion condition on the ancilla graph.
    \item \textbf{QLDPC preservation.} When the input is QLDPC, the construction must keep qubit participation and check weight bounded. In the graph, these quantities are controlled by vertex degree, path-matching load, cycle congestion, and cycle length.
    \item \textbf{Resource minimization.} Subject to the first three requirements, the compiler should minimize the ancilla qubits and checks introduced by the construction.
\end{enumerate}
The first three objectives are feasibility requirements; the fourth ranks feasible constructions by their resource cost. They are coupled rather than enforced by independent sequential passes: adding an edge can improve expansion, change the cycle space and its congestion, and increase the number of ancilla qubits at the same time.

Although the synthesized ancilla system is CSS, the base code $\mathcal{Q}$ need not be. The compiler reads the stabilizer presentation of $\mathcal{Q}$ only to determine check support and local anti-commutation with the target logical operator. For a cross-code measurement, two synthesized ancilla systems are connected by an adapter or bridge~\cite{swaroop2024universal}. We formulate the core problem for one ancilla system because the same IR is used for each side of that construction.

\subsection{IR Syntax and Denotation}\label{Sec:IRSyntax}

Let $(\mathcal{Q},L)$ be a code surgery instance, and let $\mathcal{S}=\langle S_1,\ldots,S_m\rangle$ be the stabilizer group of $\mathcal{Q}$. We write $\mathsf{Chk}(\mathcal{Q})=\{S_1,\ldots,S_m\}$ for the finite set of checks in the input presentation. A surgery graph $G=(V,E)$ attaches to the support of $L$ through a set of ports $P\subseteq V$ and a bijection $f:\operatorname{supp}(L)\to P$. Thus the port $f(q)$ is attached to physical qubit $q$.

We use local Pauli notation in this section. If $R=R_1\otimes\cdots\otimes R_n$ is a Pauli string on the original code block, then $R(q)$ is its Pauli factor on physical qubit $q$. In particular, $S(q)$ and $L(q)$ denote the local factors of an input check and the logical operator. Each graph edge $e$ becomes an ancilla qubit; $X(e)$ and $Z(e)$ denote Pauli $X$ and $Z$ on that qubit.

An input check $S$ commutes with $L$ globally, but their local Pauli factors may anti-commute on individual qubits. Define
\[
K(S,L)=\{q\in\operatorname{supp}(L):S(q)L(q)=-L(q)S(q)\}.
\]
The set $K(S,L)$ has even cardinality because $S$ and $L$ commute. The construction repairs these local anti-commutations in pairs.

\begin{definition}[Path matching graph]\label{def:path-matching-graph}
A \emph{path matching} for $S\in\mathsf{Chk}(\mathcal{Q})$ is a set of edge-disjoint paths in $G$ whose endpoints pair the ports $f(K(S,L))$. A path-matching dictionary $\mathcal{M}$ maps each input check $S$ to the edges in its path matching, with $\mathcal{M}(S)=\emptyset$ when $K(S,L)=\emptyset$. The \emph{path matching graph} $G_0$ is the union of the edges required by these path matchings.
\end{definition}

The path matching graph is the part of the ancilla graph forced by the input instance: it records how the compiler repairs commutation before making any choices for distance, QLDPC structure, or resource reduction. Subsequent operations may add edges or layers, but they retain the paths recorded in $\mathcal{M}$. Our implementation pairs the ports in $f(K(S,L))$ by shuffling them with a fixed seed and joining consecutive entries by a single edge. Different pairings can affect the subsequent optimization, but in our benchmarks most nonempty sets $f(K(S,L))$ have two ports, and thus a unique pairing, and the rest have four; we keep the pairings fixed in all experiments.

\begin{definition}[Surgery graph IR]\label{def:surgery-graph-ir}
For a code surgery instance $(\mathcal{Q},L)$, a \emph{surgery graph IR} is a tuple
\[
\mathfrak{G}=(G,P,f,\mathcal{M},\mathcal{C}),
\]
where $G=(V,E)$ is the current ancilla graph, $P\subseteq V$ is its port set, $f:\operatorname{supp}(L)\to P$ is the port map, and $\mathcal{M}$ is the path-matching dictionary. The component $\mathcal{C}$ is a cycle representation of $G$ (Definition~\ref{def:cycle-representation}). A final IR is emitted only when $\mathcal{C}$ generates the full cycle space and has been cellulated into bounded-size cycles.
\end{definition}

The graph and port fields specify how the ancilla system attaches to $L$; $\mathcal{M}$ records the commutation repair forced by the input; and $\mathcal{C}$ records the cycle constraints maintained while the graph changes. This makes the object manipulated by the compiler the same object that denotes the generated stabilizer code.

We next define that denotation. The equations use the $X/Z$ orientation for the CSS ancilla system: vertex checks use $X$ on ancilla edges and cycle checks use $Z$. The dual orientation exchanges $X$ and $Z$; the local factors $S(q)$ and $L(q)$ of the base code may be arbitrary Paulis.

Each vertex $v\in V$ induces a \emph{vertex check}
\[
A_v=
\begin{cases}
\displaystyle\left(\prod_{e\ni v}X(e)\right)L(q), & \text{if }v=f(q)\in P,\\[0.7em]
\displaystyle\prod_{e\ni v}X(e), & \text{otherwise.}
\end{cases}
\]
Because every graph edge has two endpoints, every ancilla operator $X(e)$ appears twice in the product of all vertex checks. Therefore
\[
\prod_{v\in V}A_v=L.
\]
The product of their measurement outcomes consequently gives the outcome of the logical measurement.

For every input check $S\in\mathsf{Chk}(\mathcal{Q})$, the dictionary supplies the repaired check
\[
\overline{S}=S\prod_{e\in\mathcal{M}(S)}Z(e).
\]
At every endpoint of a matching path, the added $Z$ contributes an ancilla anti-commutation that cancels the local anti-commutation between $S(q)$ and $L(q)$. At an internal path vertex, the path contributes two incident $Z$ operators, so their anti-commutations cancel in pairs. Finally, every represented cycle $C\in\mathcal{C}$ induces a \emph{cycle check}
\[
B_C=\prod_{e\in C}Z(e).
\]

\begin{definition}[Surgery denotation]\label{def:surgery-denotation}
For $\mathfrak{G}=(G,P,f,\mathcal{M},\mathcal{C})$, the denotation $\llbracket\mathfrak{G}\rrbracket_{\mathcal{Q},L}$ is the deformed code on the original qubits of $\mathcal{Q}$ and one ancilla qubit per edge of $G$, generated by
\[
\widetilde{\mathcal{S}}=
\left\langle
\{A_v:v\in V\}
\cup\{\overline{S}:S\in\mathsf{Chk}(\mathcal{Q})\}
\cup\{B_C:C\in\mathcal{C}\}
\right\rangle.
\]
\end{definition}

\begin{proposition}[Local commutation repair]\label{prop:local-commutation-repair}
If $\mathcal{M}$ maps every $S\in\mathsf{Chk}(\mathcal{Q})$ to a valid path matching for $f(K(S,L))$, then all vertex checks, repaired input checks, and cycle checks in Definition~\ref{def:surgery-denotation} mutually commute.
\end{proposition}

\begin{proof}
Two vertex checks commute because they use the same Pauli type on every shared ancilla qubit and, by bijectivity of $f$, act on distinct original qubits. For a repaired check $\overline{S}$, the two possible anti-commutations at every endpoint of a matching path cancel, while internal path vertices have even path incidence. A cycle check commutes with every vertex check because a cycle has even incidence at every vertex. The repaired input checks and cycle checks use only $Z$ on shared ancilla qubits, and the original input checks commute with one another.
\end{proof}

Path matching therefore discharges the first feasibility requirement. The remaining operations do not change the target logical operator or invalidate an existing path matching; they change the graph and cycle representation to meet the distance and QLDPC requirements with fewer resources.

\subsection{IR Transformations and Their Effects}\label{Sec:IRTransformations}

The compiler evolves the whole IR tuple rather than manipulating an isolated graph. We use the same functional notation for four IR transformations: $\mathsf{expand}_{e}$ adds one edge instance, $\mathsf{thicken}_{\ell}$ forms $\ell$ levels, $\mathsf{decongest}_{\Pi}$ places cycles across those levels, and $\mathsf{cellulate}_{w}$ bounds cycle length. Each takes the current IR and produces an updated IR. Expansion additionally threads a spanning forest as auxiliary synthesis state because its cycle update depends on the edge-insertion order.

\begin{definition}[Expansion transition (adding edge)]\label{def:expansion}
Let
\[
\mathfrak{G}_i=(G_i,P,f,\mathcal{M},\mathcal{C}_i),
\qquad G_i=(V,E_i),
\]
be the current IR, and let $T_i$ be the spanning forest maintained as auxiliary synthesis state. For one selected edge instance $e_i=\{u,v\}$, an \emph{expansion transition} is
\[
(\mathfrak{G}_{i+1},T_{i+1})
=\mathsf{expand}_{e_i}(\mathfrak{G}_i,T_i).
\]
The transition sets $G_{i+1}=(V,E_i\uplus\{e_i\})$ and
$\mathfrak{G}_{i+1}=(G_{i+1},P,f,\mathcal{M},\mathcal{C}_{i+1})$.
If $u$ and $v$ belong to different components of $T_i$, then
\[
\mathcal{C}_{i+1}=\mathcal{C}_i,
\qquad T_{i+1}=T_i\uplus\{e_i\}.
\]
Otherwise, let $P_{T_i}(u,v)$ be their unique path in $T_i$ and let
$r_i=P_{T_i}(u,v)\uplus\{e_i\}$. The transition sets
\[
\mathcal{C}_{i+1}=\mathcal{C}_i\cup\{r_i\},
\qquad
T_{i+1}=(T_i\uplus\{e_i\})\setminus\{a_i\},
\]
where $a_i\in P_{T_i}(u,v)$ is selected by the cycle-maintenance rule in Section~\ref{Sec:DynamicCycleTrackingAlgorithm}. The forest $T_i$ is synthesis state and is not part of the emitted IR.
\end{definition}

A sequence of selected edge instances is processed by composing these transitions in selection order. The transitions need not commute: changing the insertion order or the choices of $a_i$ can change the maintained forest, generated cycles, and congestion. Each transition retains $P$, $f$, and $\mathcal{M}$. If $\mathcal{C}_i$ generates the cycle space of $G_i$, then $\mathcal{C}_{i+1}$ generates the cycle space of $G_{i+1}$: an edge joining two components does not change the cycle-space dimension, whereas an edge within one component adds its fundamental cycle.

Intuitively, one expansion transition adds an alternative route through the ancilla graph. An edge crossing a bottleneck makes it more expensive to separate ports, which moves the graph toward the relative-expansion condition required for distance preservation. The trade-off is immediate in the IR: the new edge is another ancilla qubit, and a new non-tree edge creates another cycle constraint that later operations must place and cellulate.

\begin{definition}[Thickening transformation]\label{def:thickening}
Let $\mathfrak{G}=(G,P,f,\mathcal{M},\mathcal{C})$. For $\ell\ge 1$, let $J_\ell$ be the path graph on $[\ell]=\{1,\ldots,\ell\}$. The \emph{thickening transformation} is
\[
\mathfrak{G}_{\ell}
=\mathsf{thicken}_{\ell}(\mathfrak{G})
=(G_\ell,P_\ell,f_\ell,\mathcal{M}_\ell,\mathcal{C}_{\ell}),
\qquad G_\ell=G\square J_\ell.
\]
The ports, path matchings, and represented cycles are embedded in the first level: $P_\ell=P\times\{1\}$, $f_\ell(q)=(f(q),1)$, $\mathcal{M}_\ell(S)$ is the first-level copy of $\mathcal{M}(S)$, and $\mathcal{C}_\ell$ contains the first-level copy of each $C\in\mathcal{C}$. It also contains the rectangular cycle $((u,r),(v,r),(v,r+1),(u,r+1))$ for every edge $\{u,v\}\in E$ and adjacent levels $r,r+1$. If $\mathcal{C}$ generates the cycle space of $G$, these cycles generate the cycle space of $G_\ell$.
\end{definition}

Thickening replaces one route through the base graph by $\ell$ horizontal copies connected vertically, which raises the relative expansion seen by the ports and creates levels for separating overlapping cycle checks. The thickened graph has $\ell|V|$ vertices and $\ell|E|+(\ell-1)|V|$ edges, so every extra level directly becomes additional checks and ancilla qubits.

\begin{definition}[Decongestion transformation]\label{def:decongestion-transformation}
Let $\mathfrak{G}_\ell=\mathsf{thicken}_\ell(\mathfrak{G})=(G_\ell,P_\ell,f_\ell,\mathcal{M}_\ell,\mathcal{C}_\ell)$ for $\mathfrak{G}=(G,P,f,\mathcal{M},\mathcal{C})$, and let $\Pi=(\Pi_1,\ldots,\Pi_\ell)$ be an ordered partition of the base cycles $\mathcal{C}$, permitting empty parts. The \emph{decongestion transformation} is
\[
\mathfrak{G}_{\ell,\Pi}
=\mathsf{decongest}_{\Pi}(\mathfrak{G}_\ell)
=(G_\ell,P_\ell,f_\ell,\mathcal{M}_\ell,\mathcal{C}_{\ell,\Pi}).
\]
The cycle representation $\mathcal{C}_{\ell,\Pi}$ contains the level-$r$ copy of each $C\in\Pi_r$ and retains all rectangular cycles from $\mathcal{C}_\ell$. It has the same $\mathbb{F}_2$-span as $\mathcal{C}_\ell$.
\end{definition}

Decongestion changes where cycle checks act, not the layered graph or the constraints they generate. Placing overlapping cycles on different levels lowers cycle-check load. Once $\ell$ is fixed, this transformation adds no qubits or checks, although a placement requiring more groups can force the compiler to choose a larger $\ell$.

\begin{definition}[Cellulation transformation]\label{def:cellulation}
Let $\mathfrak{G}=(G,P,f,\mathcal{M},\mathcal{C})$. For a target check-weight bound $w\ge 3$, the \emph{cellulation transformation} is
\[
\mathfrak{G}'=\mathsf{cellulate}_{w}(\mathfrak{G})
=(G',P,f,\mathcal{M},\mathcal{C}').
\]
It keeps the vertex set fixed and forms $G'=(V,E\uplus E_{\mathrm{cell}})$ by adding a multiset $E_{\mathrm{cell}}$ of edge instances. It replaces every long cycle in $\mathcal{C}$ by cycles of length at most $w$ such that the original cycle is their $\mathbb{F}_2$-sum, and $\mathcal{C}'$ generates the cycle space of $G'$. The ports, port map, and path matchings are unchanged.
\end{definition}

Cellulation turns one high-weight cycle check into several bounded-weight checks whose product recovers the original constraint. It enforces bounded cycle-check weight, while the added edge instances introduce ancilla qubits and the replacement cycles introduce checks. A larger bounded value of $w$ permits larger cells and therefore fewer added qubits and checks; Section~\ref{Sec:IncorporatingDegreeConstraintsintoPracticalCodeDesign} exploits this trade-off under the QLDPC budget.

These four transformations evolve the candidate IR rather than preserve a fixed denotation.  They maintain three structural invariants: (I1) the attachment $(P,f)$ to $L$ is retained or embedded; (I2) every $\mathcal{M}(S)$ remains a valid path matching; and (I3) $\mathcal{C}$ remains a complete cycle representation, as expansion adds each new fundamental cycle, thickening adds the rectangular cycles, decongestion preserves the span, and cellulation replaces long cycles by bounded cells spanning the augmented cycle space. Their effects on distance, degree, and resources differ, which is why the compiler must choose their parameters jointly.

\subsection{Feasibility and Resource Objective}\label{Sec:Feasibility}

We now collect the first three objectives as conditions on the final IR and make the fourth objective explicit.

\textbf{Distance preservation.}
A construction is \emph{distance preserving} when the deformed code has distance at least the target distance $d$ of the input instance. Prior graph-based surgery constructions~\cite{williamson2024low,he2025extractors} establish this property when the final ported graph satisfies
\[
\beta_d(G,P)\ge 1.
\]
Intuitively, a cut that removes support from the original logical operator must then cross at least as many ancilla edges as the amount of support it removes, so the cut cannot dress the logical operator into one of lower weight.

\textbf{QLDPC preservation.}
When $\mathcal{Q}$ is QLDPC, the deformed code should retain bounded qubit participation and bounded check weight. Vertex degree controls vertex-check weight, cycle length controls cycle-check weight, and cycle congestion contributes directly to the number of checks acting on an edge qubit. Every edge qubit participates in exactly two endpoint vertex checks, while path-matching repair and the rectangular checks introduced by thickening add only a fixed number of further incidences. Cellulation preserves the tracked cycle load on existing edges and gives each new local edge bounded load. Cycle-basis congestion is therefore the only contribution that can grow with the synthesized graph and is the practical degree bottleneck. In the inherited construction, each matched pair is joined by its own edge instance. Hence
$
|\mathcal{M}(S)|=\frac{|K(S,L)|}{2}\leq\frac{\operatorname{wt}(S)}{2},
\
\operatorname{wt}(\overline{S})\leq\frac{3}{2}\operatorname{wt}(S).
$
Thus bounded input-check weight gives bounded repaired-check weight, and each matching-edge instance receives only one additional incidence from path-matching repair. Bijectivity of $f$ also adds at most one vertex check to each original qubit.

We use two numeric degree budgets: $d_q$ bounds the variable cycle-basis load on any ancilla edge, and $d_c$ bounds the weight of any generated vertex or cycle check. Thus every ancilla edge participates in at most $d_q$ cycle-basis checks and $O(1)$ other checks, while every generated vertex or cycle check acts on at most $d_c$ qubits.

If the input code is QLDPC and $d_q,d_c=O(1)$ independently of code size, these requirements keep the deformed code QLDPC. They also prevent a purported resource reduction from hiding unbounded qubit participation or check weight.

\begin{definition}[Feasible surgery graph]\label{def:feasible-surgery-graph}
Let $(\mathcal{Q},L)$ have target distance $d$, and let $(d_q,d_c)$ be degree budgets. A final surgery graph IR $\mathfrak{G}=(G,P,f,\mathcal{M},\mathcal{C})$ is \emph{feasible} if:
\begin{enumerate}
    \item for every $S\in\mathsf{Chk}(\mathcal{Q})$, $\mathcal{M}(S)$ is a valid path matching for $f(K(S,L))$ with $|\mathcal{M}(S)|\leq\operatorname{wt}(S)/2$, and every edge instance belongs to at most one $\mathcal{M}(S)$;
    \item the final ported graph satisfies $\beta_d(G,P)\ge 1$;
    \item $\mathcal{C}$ is a set of bounded-size cycles that generates the cycle space of $G$; and
    \item no ancilla edge exceeds cycle-basis load $d_q$, and every generated vertex or cycle check has weight at most $d_c$.
\end{enumerate}
\end{definition}

\begin{theorem}[Correctness of surgery graphs]\label{thm:surgery-correctness}
Let $\mathfrak{G}$ be a feasible surgery graph for $(\mathcal{Q},L)$ under degree budgets $(d_q,d_c)$, and let $\widetilde{\mathcal{Q}}=\llbracket\mathfrak{G}\rrbracket_{\mathcal{Q},L}$. Then $\widetilde{\mathcal{Q}}$ is a well-defined stabilizer code in which $L$ is measurable as a stabilizer. Moreover, the graph expansion condition gives the distance-preservation guarantee of graph-based code surgery. Each ancilla qubit participates in at most $d_q$ cycle-basis checks and $O(1)$ other checks, and each generated vertex or cycle check has weight at most $d_c$.
\end{theorem}

The theorem follows the graph-based surgery construction of~\cite{williamson2024low,he2025extractors}; Appendix~\ref{proof:surgery-correctness} gives a proof sketch in the notation of our IR. The path-matching condition gives commutation by Proposition~\ref{prop:local-commutation-repair}; the relative-expansion condition supplies the inherited distance guarantee; and the degree budgets give the QLDPC bounds. Thus \myCompilerNameSpace searches only among constructions with these guarantees rather than relying on post-compilation distance verification.

\begin{remark}[Consolidated correctness argument]\label{rem:consolidated}
Definition~\ref{def:feasible-surgery-graph} characterizes the valid surgery structures, and the synthesis establishes each of its conditions in one of two ways. Conditions~1 and~3 are invariants: the initial path-matching IR satisfies (I1)--(I3), and every transformation preserves them, so every intermediate IR already denotes a valid deformed code in which $L$ is a stabilizer and all gauge degrees of freedom are fixed. Conditions~2 and~4 are not preserved by individual transformations; they hold for the final IR by the choice of parameters. The stopping rule thickens the base graph by $\ell\geq\lceil 2/\lambda_2\rceil\geq 1/\beta$ levels, so the relative expansion reaches $\ell\beta\geq 1$, and cellulation only adds edges and cannot decrease any cut; the degree-constrained partition and $d_c$-gon cellulation bound every cycle load by $d_q$ and every cycle check by $d_c$ (Section~\ref{Sec:IncorporatingDegreeConstraintsintoPracticalCodeDesign}). Hence \myCompilerNameSpace obtains the guarantees of Theorem~\ref{thm:surgery-correctness} by construction.
\end{remark}

For a feasible final IR with $G=(V,E)$, every edge is an introduced ancilla qubit. Vertex checks and cycle checks are the newly introduced checks; repaired input checks replace their original counterparts. The two resource objectives are therefore
\[
\mathsf{qubits}(\mathfrak{G})=|E|,
\qquad
\mathsf{checks}(\mathfrak{G})=|V|+|\mathcal{C}|,
\]
plus the adapter or bridge contribution reported separately for cross-code measurements. The compiler minimizes both quantities subject to feasibility rather than assigning an arbitrary conversion factor between qubits and checks.

\subsection{Coupled Synthesis}\label{Sec:CoupledSynthesis}

The transformations above describe distinct effects, but \myCompilerNameSpace does not execute them as independent optimization passes over completed intermediate graphs. It maintains the current graph, path-matching dictionary, cycle representation, spanning forest, cycle partition, and edge-load counters as shared synthesis state. Expansion transitions are applied in sampled order; after each accepted edge, the cycle representation and its eventual placement are updated immediately. The stopping rule compares the thickening needed for expansion with the thickening needed for decongestion. Cellulation is then chosen under the same degree budgets that constrained cycle placement.

This coupling explains why a fixed construction can be wasteful. More expansion edges can reduce the number of layers needed for distance preservation, but they consume qubits and create cycles. More layers can reduce cycle-check load, but they replicate the entire base graph. Larger bounded cells can reduce cellulation overhead, but only up to the permitted check weight. The synthesis problem is therefore to choose an ordered sequence of $\mathsf{expand}$ transitions and final parameters $(\ell,\Pi,w)$ for $\mathsf{thicken}$, $\mathsf{decongest}$, and $\mathsf{cellulate}$ that minimizes $|E|$ and $|V|+|\mathcal{C}|$ subject to feasibility. Solving this problem exactly is intractable already at its expansion step: certifying that a candidate graph meets an expansion threshold requires its Cheeger constant, which is NP-hard to compute~\cite{mohar1989isoperimetric,kaibel2004expansion}. \myCompilerNameSpace therefore certifies expansion with the efficiently computable spectral bound $\beta(G)\geq\lambda_2(G)/2$ (Section~\ref{Sec:OptimizingExpanderConstructionviaConditionedRandomization}) and searches over edge insertions by conditioned sampling in $O(n^4)$ worst-case time in the ancilla-graph size (Section~\ref{Sec:Scalability}). This separates synthesis from optimization: every synthesis run returns a final surgery graph with the guarantees of Section~\ref{Sec:Feasibility}, whereas the search for a small one is heuristic.

The next three sections develop this coupled procedure. Section~\ref{Sec:OptimizingExpanderConstructionviaConditionedRandomization} reuses the path-matching graph when proposing expansion edges. Section~\ref{Sec:BalancingExpansionandCongestionthroughDynamicCycleTracking} updates the cycle representation, congestion information, and thickening requirement as those edges are added. Section~\ref{Sec:IncorporatingDegreeConstraintsintoPracticalCodeDesign} uses the degree budgets when placing and cellulating the tracked cycles. Together, these mechanisms reduce resource cost while satisfying the three final feasibility requirements above.

%% file: tex/4-OptimizingExpanderConstruction.tex
\section{Optimizing Expander Construction via Conditioned Randomization} \label{Sec:OptimizingExpanderConstructionviaConditionedRandomization}

In Section~\ref{Sec:SurgeryGraphIR}, we formulated code surgery synthesis as a constrained optimization problem over the surgery graph IR. This section focuses on the expansion constraint: construct an ancilla graph with enough expansion for distance preservation while adding as few ancilla qubits, i.e., graph edges, as possible.

In existing constructions, the graph is built in multiple stages. As illustrated on the left of Figure~\ref{fig:Pathmatching}, the construction first introduces the \emph{path-matching graph} defined in Section~\ref{Sec:SurgeryGraphIR}, which arises from enforcing commutation between the original stabilizers and the newly introduced vertex checks. This initial graph already contains a nontrivial amount of connectivity determined by the input code and logical operator.

However, prior approaches treat this graph merely as an intermediate object and construct an independent constant-degree expander on top of it (Figure~\ref{fig:Pathmatching}, top). Constant-degree expanders are used because they simultaneously provide expansion and preserve bounded check degree. The drawback is that an independent expander ignores the connectivity already present in the path-matching graph, yielding redundant edges and unnecessary overhead.

\subsection{Conditioned Expander Construction} \label{Sec:ConditionedExpanderConstruction}

\begin{wrapfigure}{r}{0.45\linewidth}
    \centering
    \vspace{-5pt}
    \includegraphics[width=\linewidth]{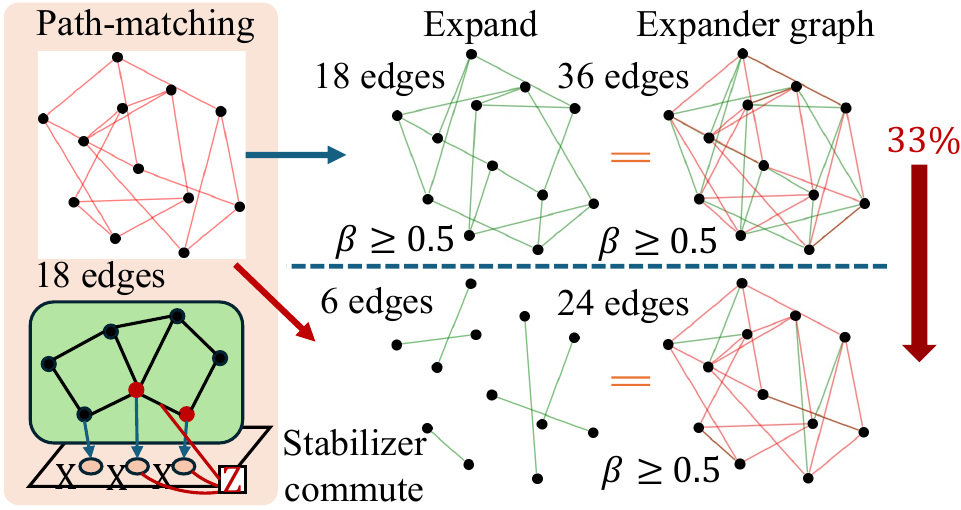}
    \vspace{-15pt}
    \caption{Conditioned expansion on the path-matching graph.}
    \label{fig:Pathmatching}
    \vspace{-5pt}
\end{wrapfigure}
We propose a low-overhead structure-aware expander construction algorithm that leverages the existing path-matching graph to achieve the expansion requirement. Instead of constructing a full expander independently, our approach progressively augments the graph by adding edges in a randomized manner while continuously monitoring its expansion. As illustrated at the bottom of Figure~\ref{fig:Pathmatching}, this process gradually improves connectivity using only a small number of additional edges until the desired expansion is achieved.

We present our construction in Algorithm~\ref{alg:ConditionedExpanderConstruction}, which leverages the uniform random graph generation algorithm of~\cite{steger1999generating}. Their approach generates $d$-regular graphs on $n$ vertices approximately uniformly at random. We use this random-regular viewpoint because random regular graphs are expanders with high probability while keeping degree bounded, exactly matching the graph-side form of the LDPC constraint.

Given a path-matching graph $G_0$, our goal in this section is to construct a base graph $G$ with as much useful expansion as possible while adding few edges. We use $\beta<1$ as a target for the base graph, not as the final correctness threshold: directly forcing the unthickened graph to reach relative expansion one would usually add too many edges and create high cycle congestion. If the augmentation budget is exhausted before reaching $\beta$, the compiler retains the achieved base expansion and compensates with additional thickening. Section~\ref{Sec:BalancingExpansionandCongestionthroughDynamicCycleTracking} explains how the achieved expansion determines the layers needed to reach the final relative-expansion threshold.

\textbf{Key Idea.}
Instead of constructing an expander independently, we \emph{condition} the construction on the existing graph $G_0$. This allows us to reuse existing edges and avoid redundant connectivity. Our construction has two phases.

The augmentation operates on the loopless multigraph model defined in Section~\ref{Sec:GraphConcepts}. It may add an edge parallel to an existing edge. The two edge instances represent distinct ancilla qubits and contribute separately to every cut, so a parallel edge can increase expansion. Degrees, adjacency-matrix entries, and edge counts below include multiplicity.

\textbf{Phase 1: Regularization to Maximum Degree.}
We first compute the maximum degree $\Delta$ over all vertices. Let $\deg_G(v)$ denote the degree of vertex $v$. We then iteratively add edges between vertices with remaining degree capacity.

At each step, we select a pair $\{u,v\}$ with $\deg_G(u) < \Delta$ and $\deg_G(v) < \Delta$, with probability proportional to $(\Delta - \deg_G(u))(\Delta - \deg_G(v))$, and add the edge to $G$. We add edges one at a time because the objective is to stop as soon as the expansion condition is met. Adding an entire batch would reduce spectral-evaluation overhead, but can overshoot the threshold and leave unnecessary ancilla qubits in the output graph.

After each addition, we evaluate the expansion of the graph.
Directly certifying expansion by checking all cuts is prohibitively expensive~\cite{chung1997spectral}. Instead, we use a spectral proxy. Let $A_G$ be the adjacency matrix of $G$ and let $D_G$ be the diagonal degree matrix. The combinatorial graph Laplacian is $L_G=D_G-A_G$. Its eigenvalues satisfy $0=\lambda_1(G)\le \lambda_2(G)\le\cdots$, and $\lambda_2(G)$ is the algebraic connectivity of $G$. By the graph Cheeger inequality~\cite{chung1997spectral},
\[
\beta_G \geq \lambda_2 / 2,
\]
under the edge-expansion normalization used here, so larger $\lambda_2$ implies better expansion.
We monitor $\lambda_2$ during construction and terminate once $\lambda_2 \geq 2\beta$.
Isolated vertices require no special case. If $\Delta>0$, they have maximum deficit and are eligible for Phase 1; if the entire path-matching graph has $\Delta=0$, Phase 1 exits immediately and Phase 2 starts adding matching layers.

\textbf{Phase 2: Degree Augmentation.}
If Phase 1 cannot reach the desired expansion, we continue augmenting the graph by adding random $1$-regular layers.
At each outer iteration, we generate a random matching and add its edges progressively, checking $\lambda_2$ after each addition. The process stops as soon as the base target is satisfied. If the number of vertices is odd, the last unmatched vertex is handled by the next sampled layer; this creates at most the bounded degree deviation discussed below. If all $\tau$ matching layers are added first, the algorithm returns the current graph; this is not a failure because subsequent thickening uses its achieved expansion.
This two-phase design balances two objectives:
\begin{itemize}
    \item Phase 1 maximally exploits existing structure,
    \item Phase 2 progressively improves expansion when structure alone is insufficient.
\end{itemize}
If $G_0$ already meets the target, the implementation uses it directly; Algorithm~\ref{alg:ConditionedExpanderConstruction} presents the nontrivial case in which augmentation is needed. In practice, we repeat the randomized construction multiple times and select the smallest resulting graph.

\begin{algorithm}[t]
\caption{Conditioned Expander Construction}
\label{alg:ConditionedExpanderConstruction}
\KwIn{Path matching graph $G_0$, base-expansion target $\beta$, maximum matching layers $\tau$}
\KwOut{Augmented base graph $G=(V,E)$}

\textbf{Initialization:} $G \leftarrow G_0$; $iteration \leftarrow 0$\;

\textbf{Phase 1: Regularization to Maximum Degree}\;
Compute $\Delta \leftarrow \max_v \deg_G(v)$\;

\While{$\exists v \in V \text{ such that } \deg_G(v) < \Delta$}{
    Define deficit $\Delta(v) \leftarrow \Delta - \deg_G(v)$ for all $v \in V$\;

    Let $S = \{\{u,v\} \mid u \neq v,\ \Delta(u)>0,\ \Delta(v)>0\}$\;

    \textbf{if} $S = \emptyset$ \textbf{then break}\;

    Select $\{u,v\} \in S$ with probability proportional to $\Delta(u)\Delta(v)$; Add edge $\{u,v\}$ to $G$;
    
    Compute $\lambda_2(G)$; \textbf{if} $\lambda_2(G) \ge 2\beta$ \textbf{then return} $G$\;
}


\textbf{Phase 2: Degree Augmentation}\;

\While{$iteration < \tau$}{
    
    $iteration \leftarrow iteration + 1$; Let $U \leftarrow V$\; 

    \While{$U \neq \emptyset$}{
        Let $S = \{\{u,v\} \mid u,v \in U,\ u \neq v\}$\;

        \textbf{if} $S = \emptyset$ \textbf{then break}\;

        Select $\{u,v\} \in S$ uniformly at random; Remove $u,v$ from $U$;  Add edge $\{u,v\}$ to $G$\;  
        Compute $\lambda_2(G)$; \textbf{if} $\lambda_2(G) \ge 2\beta$ \textbf{then return} $G$\;

        
    }

}

\Return $G$\;

\end{algorithm}

\subsection{Sampling Expanders with High Probability} \label{Sec:SamplingExpanderswithHighProbability}
Algorithm~\ref{alg:ConditionedExpanderConstruction} checks expansion after every sampled edge and may stop partway through a matching layer. For the asymptotic analysis, we instead consider an auxiliary process that suppresses this early stopping and inspect it only after complete matching layers. This auxiliary process is not a different implementation; it is a convenient lower-bound process contained in any execution that reaches the same layer boundary.
We first define the probability space for our sampling algorithm.
\begin{definition}[Random Regular Graph Models]
Let $n,r \in \mathbb{N}$ with $nr$ even.

\textbf{(1) Simple $r$-regular graph.}
Let $\Omega_{n,r}$ denote the set of all simple $r$-regular graphs on $n$ labelled vertices, and let $\mathcal{F}_{n,r} = 2^{\Omega_{n,r}}$. Let $\mathbb{U}_{n,r}$ be the uniform distribution over $\Omega_{n,r}$. We define
\[
G(n,r) : (\Omega_{n,r}, \mathcal{F}_{n,r}, \mathbb{U}_{n,r}) \to (\Omega_{n,r}, \mathcal{F}_{n,r})
\]
as a random element taking values in $\Omega_{n,r}$, corresponding to a uniformly sampled simple $r$-regular graph, which can be obtained via the configuration model conditioned on simplicity.
The configuration model~\cite{bollobas1980probabilistic,steger1999generating} creates $r$ half-edges at each of the $n$ vertices, pairs all $nr$ half-edges uniformly at random, and then identifies each vertex's half-edges with that vertex. This produces an $r$-regular multigraph; conditioning on the absence of self-loops and parallel edges yields a simple $r$-regular graph.

\textbf{(2) Special case ($r=1$).}
When $r=1$, $G(n,1)$ corresponds to a uniformly random perfect matching on $n$ vertices.

\textbf{(3) Loopless multigraph model.}
Let $G'(n,r)$ denote the random multigraph generated by the configuration model conditioned on being loopless, while allowing multiple edges. This defines a probability space $(\Omega'_{n,r}, \mathcal{F}'_{n,r}, \mathbb{P}'_{n,r})$.
\end{definition}

In this probability space, we can model the sampling process with the following definition.
\begin{definition}[Union of Random Matchings]
Let $n,r \in \mathbb{N}$ with $n$ even. Let $G(n,1)$ denote a random perfect matching on $n$ vertices with probability space $(\Omega_{n,1}, \mathcal{F}_{n,1}, \mathbb{U}_{n,1})$. 

We define the random graph $\mathcal{I}(n,r)$ as the union of $r$ independent samples of $G(n,1)$:
$\mathcal{I}(n, r) = G_1 + G_2 + \cdots + G_r$,
where $+$ denotes graph union on a common vertex set.

Equivalently, $\mathcal{I}(n,r)$ can be viewed as a random element
\[
\mathcal{I}(n, r) : \left(\Omega_{n,1}^{\otimes r}, \ \mathcal{F}_{n,1}^{\otimes r}, \ \mathbb{U}_{n,1}^{\otimes r}\right)
\to (\Omega'_{n,r}, \mathcal{F}'_{n,r}),
\]
where $\Omega'_{n,r}$ denotes the space of loopless multigraphs on $n$ vertices with degree $r$.
\end{definition}

Finally, the following theorem states that a constant number of complete random-matching layers gives a high-probability expansion guarantee. It explains why the sampler is expected to encounter expanding graphs; finite compiler runs still use the measured value of $\lambda_2(G)$ to determine when to stop and how much thickening remains necessary.

\begin{theorem}[Expansion via Reduction to $\mathcal{I}(n,r)$] \label{th:ExpansionviaReduction}
Let $n$ be even and fix $r\ge 3$. Let $G^{(r)}$ be the graph obtained with early stopping disabled, after completing Phase 1 of Algorithm~\ref{alg:ConditionedExpanderConstruction} and adding $r$ complete, independent Phase-2 matching layers. As $n\to\infty$, $G^{(r)}$ is an expander with probability $1-o(1)$.
\end{theorem}

\begin{proof}
See Appendix~\ref{proof:ExpansionviaReduction}.
\end{proof}

The theorem concerns the unstopped process only at matching-layer boundaries. The actual algorithm can return earlier, including partway through a layer, as soon as its finite spectral test reaches the base target. Conversely, reaching the budget $\tau$ without meeting that target does not reject the sampled graph: the algorithm returns its current base graph, and the thickening calculation in Section~\ref{Sec:BalancingExpansionandCongestionthroughDynamicCycleTracking} supplies the remaining expansion.

\begin{remark}
In certain corner cases (e.g., due to parity constraints), a perfectly $d$-regular graph may not exist. In such cases, we allow a small deviation from regularity by assigning one vertex a degree higher by one than the others. Since this modification affects only a constant number of vertices, it does not alter the asymptotic expansion guarantees.
\end{remark}

%% file: tex/5-BalancingExpansionandCongestion.tex
\section{Balancing Expansion and Congestion} 
\label{Sec:BalancingExpansionandCongestionthroughDynamicCycleTracking}

In the previous section, we constructed an expander graph with reduced edge redundancy. However, this construction leaves two important issues unresolved. First, the achieved base expansion is bounded by a constant $\beta < 1$, and must be further boosted to meet the relative-expansion threshold needed for distance preservation. Second, the resulting graph may admit a cycle basis with nontrivial congestion, rather than the desired bounded congestion, because cycles can overlap significantly on certain edges.

\subsection{Graph Thickening and its Overhead}

\begin{wrapfigure}{r}{0.48\linewidth}
    \centering
    \vspace{-5pt}
    \includegraphics[width=\linewidth]{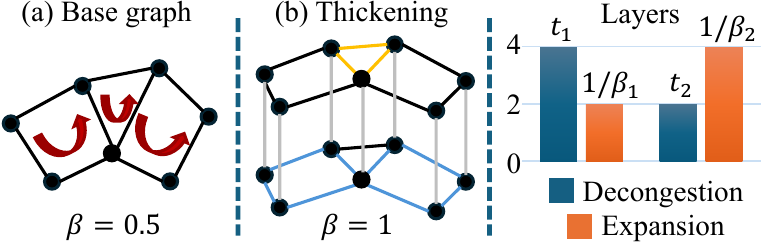}
    \vspace{-15pt}
    \caption{Graph thickening trade-off.}
    \label{fig:GraphThickenin}
    \vspace{-5pt}
\end{wrapfigure}
Graph thickening provides a mechanism to simultaneously boost expansion and reduce congestion. Starting from a base graph with Cheeger constant $\beta$ (e.g., $\beta = 0.5$ in Figure~\ref{fig:GraphThickenin}(a)), it constructs multiple copies of the graph and connects them across layers. By using more edges (i.e., more physical qubits in the synthesized ancilla system), this method increases the relative expansion to $\ell \cdot \beta$ after $\ell$ layers, while allowing cycles to be distributed across layers (e.g., the blue and yellow cycles in Figure~\ref{fig:GraphThickenin}(b)), thereby reducing overlap.

To quantify the amount of decongestion required, we partition the cycle basis of the base graph into edge-disjoint groups using a greedy procedure, and let $t$ denote the number of such groups. Intuitively, $t$ captures the number of layers needed to fully separate overlapping cycles. As a result, the total number of layers required by thickening is given by $\ell = \max(t, 1/\beta)$.

However, this introduces a fundamental trade-off, as shown on the right of Figure~\ref{fig:GraphThickenin}. A larger $\beta$ reduces the number of layers required for boosting expansion, but leads to higher congestion and thus requires more layers for decongestion. Conversely, a smaller $\beta$ simplifies decongestion but increases the number of layers needed to achieve sufficient expansion.

This observation motivates the main technique of this section: instead of fixing $\beta$ in advance, the compiler monitors congestion during expander construction and identifies a trade-off point at which the layers required for expansion boosting and decongestion are balanced.

\subsection{Dynamic Cycle Tracking Algorithm} \label{Sec:DynamicCycleTrackingAlgorithm}

We again start with a path-matching graph $G_0$ and apply the Decongestion Lemma~\cite{freedman2021building} to construct a cycle basis $\mathcal{R}_0$ for $G_0$. We also compute a spanning forest $T_0$ of $G_0$.
The Decongestion Lemma is used here as a graph-theoretic black box: given a graph, it returns a cycle basis whose edge congestion is polylogarithmically bounded. Combined with the partitioning result used in graph-based surgery~\cite{he2025extractors}, this basis can be organized into a polylogarithmic number of groups, each consisting of edge-disjoint cycles. This is the starting point that our online construction maintains as new edges are added.

Since our goal in this section is to monitor how the cycle basis evolves as new edges are added, we introduce a cycle-tracking algorithm (Algorithm~\ref{alg:cycle-tracking}). Given the current cycle basis $\mathcal{R}_i$, the spanning forest $T_i$, and a newly added edge $e_i = (u,v)$, the algorithm outputs an updated cycle basis $\mathcal{R}_{i+1}$ and a new spanning forest $T_{i+1}$. The diameter of $T_i$ matters because each new non-tree edge creates a fundamental cycle whose length is exactly the tree-path length between its endpoints plus one. Short tree paths therefore give short generated cycles, and short cycles are easier to partition with low congestion.

\begin{algorithm}[t]
\caption{Dynamic Cycle Basis Maintenance}
\label{alg:cycle-tracking}

\KwIn{Cycle basis $\mathcal{R}_i$, spanning forest $T_i$, new edge $e_i=(u,v)$}
\KwOut{Updated cycle basis $\mathcal{R}_{i+1}$, spanning forest $T_{i+1}$, optionally a new cycle $r_{i+1}$}

\If{$u$ and $v$ are not connected in $T_i$}{
    $T_{i+1} \gets T_i \cup \{e_i\}$; $\mathcal{R}_{i+1} \gets \mathcal{R}_i$;   $ r_{i+1} \gets \emptyset$\;
}
\Else{
    Let $P_{T_i}(u,v)$ be the unique path between $u$ and $v$ in $T_i$; $r_{i+1} \gets P_{T_i}(u,v) \cup \{e_i\}$\;
    $\mathcal{R}_{i+1} \gets \mathcal{R}_i \cup \{r_{i+1}\}$; $T' \gets T_i \cup \{e_i\}$\;
    Select $a_i$ uniformly at random from $P_{T_i}(u,v)$; $T_{i+1} \gets T' \setminus \{a_i\}$\;
}

\Return{$(\mathcal{R}_{i+1}, T_{i+1}, r_{i+1})$}\;

\end{algorithm}

Algorithm~\ref{alg:cycle-tracking} maintains the cycle basis dynamically. If $u$ and $v$ lie in different connected components of $T_i$, the edge $e_i$ is added to the forest and the cycle basis remains unchanged. Otherwise, $e_i$ induces a unique fundamental cycle formed by the path between $u$ and $v$ in $T_i$, which is appended to the cycle basis. To preserve the forest structure and obtain the bound below, the algorithm removes one edge uniformly at random from the corresponding tree path.

\textbf{Maintaining Low Congestion.}
We now analyze the congestion of the cycle basis constructed by Algorithm~\ref{alg:cycle-tracking}, and present the following theorem.

\begin{theorem}[Low Congestion under Short Paths] \label{th:LowCongestionunderShortPaths}
Consider a sequence of calls to Algorithm~\ref{alg:cycle-tracking} on an $n$-vertex graph with at most $m$ edge instances, without externally replacing the maintained spanning forest. If every tree path used to generate a cycle has length $O(\log n)$, then the cycles added by these calls have congestion $O(\log n\cdot\log m)$ with probability at least $1/2$.
\end{theorem}

\begin{proof}

See Appendix~\ref{proof:LowCongestionunderShortPaths}.
\end{proof}

The complete basis is the union of the initial basis $\mathcal{R}_0$ and the cycles added online, so its congestion is the congestion of $\mathcal{R}_0$ plus the bound in Theorem~\ref{th:LowCongestionunderShortPaths}. The theorem is conditional on short tree paths and does not apply to an arbitrary spanning forest, whose paths may have length $O(n)$.

In our setting, the underlying graph is an expander, which has diameter $O(\log n)$~\cite{hoory2006expander}. Recomputing the spanning forest via BFS therefore returns it to a low-diameter structure. Such a rebuild starts a new forest-maintenance epoch because it may place previously removed edges back in the forest; Theorem~\ref{th:LowCongestionunderShortPaths} applies separately within each epoch. Between rebuilds, the forest can drift and its paths can become longer. The compiler therefore tracks the realized cycle partition and uses its actual layer count $t$, rather than relying on the asymptotic congestion bound, when choosing the final thickening factor.

\textbf{Independence of the Cycle Basis.}
Since the spanning forest is dynamically modified (e.g., due to resets), it is not immediate that the cycles produced by Algorithm~\ref{alg:cycle-tracking} still form a valid cycle basis. In particular, we need to ensure that these updates do not destroy linear independence.

\begin{theorem}[Independence of the Constructed Cycle Basis] \label{Sec:IndependenceoftheConstructedCycleBasis}
Let $r_1, \ldots, r_k$ denote the cycle basis obtained by applying the Decongestion Lemma to the path-matching graph, and let $r_{k+1}, \ldots, r_l$ be the additional cycles generated by Algorithm~\ref{alg:cycle-tracking} during the expander construction process. Then the set $\{r_1, \ldots, r_l\}$ is linearly independent.
\end{theorem}

\begin{proof}


See Appendix~\ref{proof:IndependenceoftheConstructedCycleBasis}.
\end{proof}

\subsection{Dynamically Computing Decongestion Layers} \label{Sec:DynamicallyComputeDecongestionLayer}

After updating the cycle basis $\mathcal{R}_{i+1}$, we compute the layer $t_{i+1}$ for decongestion at this step. We attempt to place the new cycle into an existing partition group without overlap; if this is impossible, we create a new group (see Algorithm~\ref{alg:incremental-partition}). 

\begin{algorithm}[t]
\caption{Dynamic Greedy Cycle Partition}
\label{alg:incremental-partition}

\KwIn{Existing partition $\mathcal{R}_i = \bigcup_{k=1}^{t_i} P_k$, new cycle $r_{i+1}$}
\KwOut{Updated partition $\mathcal{R}_{i+1} = \bigcup_{k=1}^{t_{i+1}} P_k$, decongestion layer $t_{i+1}$}

\For{$k \leftarrow 1$ \KwTo $t_i$}{
    \If{$r_{i+1}$ does not overlap with any cycle in $P_k$}{
        $P_k \leftarrow P_k \cup \{r_{i+1}\}$; \Return{$ \bigcup_{k=1}^{t_i} P_k$, $t_i$}\;
    }
}

$P_{t_i+1} \leftarrow \{r_{i+1}\}$; \Return{$\bigcup_{k=1}^{t_i+1} P_k$, $t_{i}+1$}\;

\end{algorithm}

\textbf{Thickening Layer Upper Bound.}
We bound the decongestion layer $t$ by reducing our Dynamic Greedy Cycle Partition to a known Greedy Partition Algorithm (Corollary~15 in~\cite{he2025extractors}).
The offline Greedy Partition Algorithm processes an ordered list of cycles and places each cycle into the first existing group whose cycles are all edge-disjoint from it; if no such group exists, it creates a new group. This known greedy procedure provides an upper bound on the thickening layer $t$, which is given by the product of the average cycle length and the cycle basis congestion.
We show that applying our progressive partition procedure incrementally is functionally equivalent when combined with our congestion-aware expander construction (later in Section~\ref{Sec:FindingtheExpansion-CongestionTrade-offPoint}).

\begin{lemma} \label{lemma:incre}
The progressive partition procedure (Algorithm~\ref{alg:incremental-partition}) produces the same partition as the Greedy Partition algorithm (see Corollary~15 in~\cite{he2025extractors}) applied to the full set of cycles.
\end{lemma}

\begin{proof}




See Appendix~\ref{proof:incre}.
\end{proof}

Under the short-path condition of Theorem~\ref{th:LowCongestionunderShortPaths}, and when no external forest rebuild occurs during the analyzed sequence, the random-removal variant inherits the upper bound $O((\log n)^3)$ on the decongestion layer $t$ from the low-congestion cycle basis and the greedy partition bound.
We highlight that this equivalence is nontrivial. A naïve approach would recompute the greedy partition after each insertion, incurring significant overhead. In contrast, our method updates the partition incrementally by only checking the newly added cycle against existing groups, naturally accommodating our dynamic cycle tracking algorithm and reducing the complexity by a factor of $O(n)$ without changing the result.

\subsection{Finding Expansion--Congestion Balance Point} \label{Sec:FindingtheExpansion-CongestionTrade-offPoint}

Using the procedure for tracking the decongestion layer after each edge insertion, we obtain an algorithm (Algorithm~\ref{alg:CongestionAwareExpanderConstruction} in Appendix~\ref{sec:CongestionAwareExpanderConstruction}) that automatically identifies the expansion--congestion balance point during expander construction.
Leveraging the two techniques above, it augments Algorithm~\ref{alg:ConditionedExpanderConstruction} by, after adding each edge, checking the number of layers required for boosting expansion, $\lceil 2/\lambda_2(G_i) \rceil$, and comparing it with the number of layers required for decongestion, $t_i$. Since the former is monotonically decreasing while the latter is monotonically increasing, we terminate once the latter exceeds the former.
This avoids continuing to add edges after extra expansion no longer reduces the final thickening factor.
After thickening by this factor, the expansion side of the stopping rule gives the relative-expansion condition $\beta_d(G,P)\ge 1$ used in Definition~\ref{def:feasible-surgery-graph}. Thus the compiler emits a graph only after the finite construction satisfies the graph condition needed for distance preservation; the asymptotic sampling result from Section~\ref{Sec:SamplingExpanderswithHighProbability} is used to justify that such sampled graphs are likely to be found.

%% file: tex/6-IncorporatingDegreeConstraints.tex
\section{Incorporating Degree Constraints}
\label{Sec:IncorporatingDegreeConstraintsintoPracticalCodeDesign}

Finally, we incorporate the LDPC degree budgets $(d_q,d_c)$ from the synthesis problem. These budgets are not a license to create dense checks: they bound the variable cycle-check participation and generated check weight in the deformed code. The remaining participation of an ancilla edge is a fixed constant, so these budgets keep its total degree bounded as well. The standard construction uses more restrictive choices for ease of analysis, namely edge-disjoint cycle groups and triangular cellulation. When the target degree budget permits slightly larger bounded degrees, the compiler can spend that budget to reduce qubits and checks while still preserving LDPC structure.

\begin{wrapfigure}{r}{0.44\linewidth}
    \centering
    \vspace{-5pt}
    \includegraphics[width=\linewidth]{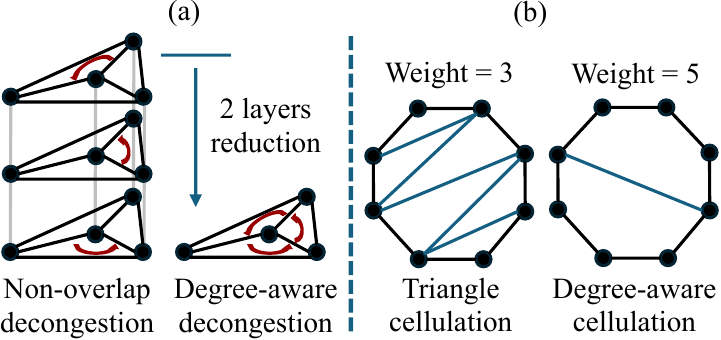}
    \vspace{-15pt}
    \caption{Degree-aware graph optimization.}
    \label{fig:Cellulation}
    \vspace{-5pt}
\end{wrapfigure}
In Figure~\ref{fig:Cellulation}(a), allowing controlled overlap under a qubit-degree budget enables more flexible cycle placement, reducing the number of layers compared to strictly non-overlapping decongestion. In Figure~\ref{fig:Cellulation}(b), degree-aware cellulation replaces uniform triangulation with higher-degree faces, reducing the number of checks (e.g., from $6$ to $2$) while respecting the check-weight bound ($5$ in this example). These observations motivate incorporating degree constraints into both decongestion and cellulation.

\subsection{Degree--Aware Cycle Basis Partition} \label{Sec:Degree--AwareCycleBasisPartition}

We first incorporate degree constraints into cycle partitioning, which is central to the decongestion process. Observe that two types of cycle partitioning are used in Algorithm~\ref{alg:CongestionAwareExpanderConstruction}: a static partition performed after constructing the path-matching graph, and a dynamic partition (Algorithm~\ref{alg:incremental-partition}) applied each time a new edge is sampled, generating a new cycle that is then assigned to a partition group.

To this end, we introduce a tracker for edge load within each cycle group $P_k$ in the partition $\mathcal{R}_i = \bigcup_{k=1}^{t_i} P_k$. This tracker records the number of cycles in which each edge participates within a given group. Since an edge is an ancilla qubit, this load is the number of variable cycle-basis checks touching that qubit.

The tracker only needs to follow this load. Each edge qubit also participates in exactly two endpoint vertex checks, while path-matching repair and the rectangular checks introduced by thickening add only a fixed number of further incidences. Cellulation preserves the tracked load on existing edges and gives each new local edge bounded load. These contributions do not grow with the code or graph size, whereas cycle-basis congestion is the practical degree bottleneck. Accordingly, $d_q$ budgets the variable cycle-basis load, and the total qubit participation is $d_q+O(1)$. When attempting to place a new cycle into a group $P_k$, we check whether its edge loads would exceed $d_q$. If all edges satisfy this condition, the cycle is assigned to $P_k$; otherwise, it is considered for the next group. We present this method in Algorithm~\ref{alg:load-incremental-partition}.

\begin{algorithm}[t]
\caption{Degree-Constrained Cycle Partition}
\label{alg:load-incremental-partition}

\KwIn{Existing partition $\mathcal{R}_i = \bigcup_{k=1}^{t_i} P_k$, new cycle $r_{i+1}$, load trackers $\{\ell_k\}_{k=1}^{t_i}$, cycle-load budget $d_q$}
\KwOut{Updated partition $\mathcal{R}_{i+1} = \bigcup_{k=1}^{t_{i+1}} P_k$, updated trackers $\{\ell_k\}_{k=1}^{t_{i+1}}$, decongestion layer $t_{i+1}$}

\For{$k \leftarrow 1$ \KwTo $t_i$}{
    \If{
        $\forall e \in r_{i+1},\ \ell_k(e) + 1 \le d_q$}{
        
        $P_k \leftarrow P_k \cup \{r_{i+1}\}$; \textbf{foreach} $e \in r_{i+1}$ \textbf{do}  $\ell_k(e) \leftarrow \ell_k(e) + 1$\;
        

        \Return{$\bigcup_{k=1}^{t_i} P_k,\; \{\ell_k\}_{k=1}^{t_i},\; t_i$}\;
    }
}

$P_{t_i+1} \leftarrow \{r_{i+1}\}$; Define a new tracker $\ell_{t_i+1}$; \textbf{foreach} $e \in r_{i+1}$ \textbf{do} $\ell_{t_i+1}(e) \leftarrow 1$\;

\Return{$\bigcup_{k=1}^{t_i+1} P_k,\; \{\ell_k\}_{k=1}^{t_i+1},\; t_i+1$}\;

\end{algorithm}

\subsection{Degree--Aware Cellulation} \label{Sec:DegreeAwareCellulation}

We now incorporate check-weight constraints into the cellulation process. In the theoretical construction~\cite{williamson2024low, he2025extractors}, cycles are cellulated into triangles using a zigzag procedure, ensuring that each cycle check acts on exactly three qubits.

To better match check-weight constraints, we generalize this construction by replacing triangular cellulation with $d_c$-gon cellulation, where $d_c$ denotes the maximum allowed check weight. The goal is to ensure that each new cycle check acts on at most $d_c$ qubits.
Given a cycle $C=(v_1,v_2,\ldots,v_\ell)$, the construction alternately cuts bounded faces from the two ends of the remaining boundary path. Every resulting cell has length at most $d_c$, and the cells sum to $C$ over $\mathbb{F}_2$. Thus replacing the check for $C$ by the cell checks preserves the complete cycle representation while directly enforcing the check-weight budget. Appendix~\ref{app:degree-aware-cellulation} gives the exact construction, resource count, and degree accounting.

%% file: tex/7-Evaluation.tex
\section{Evaluation} \label{Sec:Evaluation}

This section evaluates \myCompilerNameSpace along four dimensions aligned with the synthesis problem in Section~\ref{Sec:SurgeryGraphIR}: ancilla-system resources, the contribution of each optimization component, logical-error behavior under a decoder and noise model, and synthesis scalability. The resource and ablation results measure the quality of the generated surgery graphs, the logical-error-rate simulation checks that the optimized constructions remain comparable to the baseline rather than claiming lower error rates, and the scalability study measures whether synthesis remains usable on large logical-support instances.

\subsection{Experimental Setup} \label{Sec:ExperimentsSetup}
\begin{wraptable}[14]{R}{0.5\linewidth}
  \centering
  \vspace{-5pt}
  \caption{Benchmark information}
  \vspace{-10pt}
  {\fontsize{7.6}{8.6}\selectfont
    \setlength{\tabcolsep}{1.6pt}%
    \resizebox{0.5\textwidth}{!}{ 
    \begin{tabular}{|m{1.8cm}<{\centering}|c|c|c|}
    \hline
    Benchmark & BB code  & HP code & QT code\\
    \hline
    \multirow{5}{*}{[[n, k, d]]} & [[72, 12, 6]]  & [[450, 32, 8]] & [[81, 13, 4]]\\ \cline{2-4}
                                    & [[90, 8, 10]] & [[882, 50, 10]] & [[90, 15, 4]] \\ \cline{2-4}
                                    & [[108, 8, 10]] & [[882, 98, 8]] & [[117, 17, 4]]\\ \cline{2-4}
                                    & [[144, 12, 12]] & [[1800, 72, 14]] & [[126, 19, 4]]\\ \cline{2-4}
                                    & [[288, 12, 18]] & [[1922, 200, 12]] & [[135, 23, 4]]\\ \cline{2-4}
    \hline
    Benchmark & \multicolumn{2}{c|}{Surface code} & Color code\\
    \hline
    \multirow{5}{*}{[[n, k, d]]} & [[25, 1, 5]]  & [[441, 1, 21]] & [[19, 1, 5]]\\ \cline{2-4}
                                    & [[49, 1, 7]] & [[1681, 1, 41]] & [[37, 1, 7]] \\ \cline{2-4}
                                    & [[81, 1, 9]] & [[3721, 1, 61]] & [[61, 1, 9]]\\ \cline{2-4}
                                    & [[121, 1, 11]] & [[6561, 1, 81]] & [[91, 1, 11]]\\ \cline{2-4}
                                    & [[169, 1, 13]] & [[10201, 1, 101]] & [[127, 1, 13]]\\ \cline{2-4}
    \hline
    \end{tabular}%
     \vspace{-40pt}
    }
  }
  \label{tab:benchmark}%
  \vspace{-40pt}
\end{wraptable}
\textbf{Baselines.} Our baseline code surgery schemes include the CKBB~\cite{cohen2022low} scheme and a protocol based on gauging logical operators~\cite{williamson2024low, he2025extractors} (denoted as Gauge). These baselines provide the closest by-construction comparison for resource overhead: CKBB is a foundational general surgery construction, while Gauge follows the same graph-construction line as our compiler but without the coupled optimization mechanisms introduced in Sections~\ref{Sec:OptimizingExpanderConstructionviaConditionedRandomization}--\ref{Sec:IncorporatingDegreeConstraintsintoPracticalCodeDesign}. To enable fair comparison, we focus the main tables on by-construction baselines that can preserve distance and LDPC-style structure through the graph conditions. We additionally compare with SSIP~\cite{cowtan2024ssip}, which estimates the distance after synthesis, in Section~\ref{Sec:OverallResult}. Other works that do not guarantee distance preservation or with different design objectives are discussed in Section~\ref{Sec:RelatedWork}.

\textbf{Benchmarks.} We evaluate \myCompilerNameSpace on a variety of QEC code families of different sizes, with detailed information summarized in Table~\ref{tab:benchmark}. Here, BB denotes Bivariate Bicycle codes, HP denotes Hypergraph Product codes, and QT denotes Quantum Tanner codes. Each entry uses the $\llbracket n,k,d\rrbracket$ code-parameter notation from Section~\ref{Sec:Stabilizer Formalism}: $n$ is the number of physical qubits, $k$ is the number of logical qubits, and $d$ is the code distance.
HP codes are obtained from~\cite{aydin2025cyclic}, and QT codes are generated randomly following~\cite{perlin2023qldpc}. For QT codes, we compute the code distance via integer programming~\cite{yoder2025tour}.

\textbf{Experimental Configurations.} To illustrate the effect of each optimization step, we consider three configurations:
(1) \textsf{Exp-Opt}, which enables the conditioned expander construction process introduced in Section~\ref{Sec:OptimizingExpanderConstructionviaConditionedRandomization};
(2) \textsf{Cong-Opt}, which incorporates the technique for balancing expansion and congestion during graph construction (Section~\ref{Sec:BalancingExpansionandCongestionthroughDynamicCycleTracking});
and (3) \textsf{Full-Opt}, which further incorporates LDPC qubit-degree and check-weight constraints into the construction (Section~\ref{Sec:IncorporatingDegreeConstraintsintoPracticalCodeDesign}). \textsf{Full-Opt} is the complete \myCompilerNameSpace configuration.
All methods are run on the same input code and target logical operator for a given table entry. The \textsf{Exp-Opt}, \textsf{Cong-Opt}, and \textsf{Full-Opt} rows are cumulative configurations of the same compiler, so differences between these rows isolate the contribution of each optimization component even though the full compiler uses them through shared synthesis state.

\textbf{Metrics.} We evaluate performance using the following metrics:
(\textbf{a}) \textbf{Qubit/Check count} (denoted as \textbf{Size}): The number of ancillary qubits and checks used in constructing the code surgery system. In tables, an entry $a/b$ reports $a$ qubits and $b$ checks.
(\textbf{b}) \textbf{Qubit/Check degree} (denoted as \textbf{Deg}): The maximum qubit participation (the largest number of checks one qubit participates in) and maximum check weight (the largest number of qubits one check acts on) in the code surgery system. 
(\textbf{c}) \textbf{Logical error rate:} The simulated logical error rate for the syndrome measurement decoding pipeline of the deformed code.
(\textbf{d}) \textbf{Time}: The wall-clock synthesis time (in seconds) required to generate the code surgery system.

The first two metrics correspond directly to the cost model in Section~\ref{Sec:SurgeryGraphIR} and the feasibility conditions in Definition~\ref{def:feasible-surgery-graph}. Size measures the resource objective, while degree checks whether the compiler stayed within the degree-budget regime required for QLDPC-style constructions. We report maximum degrees explicitly because reductions obtained by increasing check weights would not be meaningful for QLDPC-oriented code surgery. Logical error rate is not used as a post-hoc correctness oracle; the distance-preservation claim follows from the graph conditions. Instead, the simulation checks how the optimized constructions behave in a complete syndrome measurement decoding pipeline. Runtime measures synthesis overhead only to evaluate the scalability of the synthesis procedure.

\textbf{Implementation.} We implement \myCompilerNameSpace in Python. All experiments are conducted on a server with a 2.4\,GHz CPU and 0.75\,TB of memory. For each input, we independently synthesize 100 valid surgery graphs and report the best one, ranking runs by thickening layers and maximum cycle load before size (Section~\ref{Sec:OverallResult} reports the distribution over all 100 runs). These runs produce 100 usable constructions; repeated synthesis is an evaluation protocol for randomized solution quality, not a prerequisite for obtaining a correct result. A user may run the compiler once for low latency or repeat it to search for a smaller construction. For a fair comparison, we set the expansion threshold to $\beta = 0.34$ and the regular-expander degree to 3 for the Gauge baseline. We use this tuned fixed threshold because it gives the smallest Gauge constructions in our experiments (i.e., the best performance): larger $\beta$ creates a larger base graph, while smaller $\beta$ requires more thickening layers to reach the final relative-expansion condition. Section~\ref{Sec:BalancingExpansionandCongestionthroughDynamicCycleTracking} analyzes this expansion-thickening trade-off.
Since the maximum qubit degrees and check weights in our benchmark codes are $12$, we use $12$ as the degree/check-weight budget of the synthesized ancilla subsystem. For Gauge, this cap applies before the ancilla system is attached to the base code. The reported \textbf{Deg} metric is measured after attachment, so incidences inherited from the base code can make the observed degree or check weight slightly exceed $12$ without violating the ancilla-subsystem cap.

\subsection{Overall Result} \label{Sec:OverallResult}
\begin{table}[!tbp]
  \centering
  \caption{Ancilla-system resources and compiler runtime}
  \vspace{-10pt}
  \label{tab:overall}%
  {\fontsize{8}{10}\selectfont
    \setlength{\tabcolsep}{2pt}%
    \renewcommand{\arraystretch}{0.95}%
    \begin{adjustbox}{max width=\textwidth}
    \begin{tabular}{|m{1.6cm}<{\centering}|c|c|c|c|c|c|c|c|c|c|c|c|c|c|c|}
    \hline
    \textbf{BB code} & \multicolumn{3}{c|}{[[72, 12, 6]] ($89\%$)} & \multicolumn{3}{c|}{[[90, 8, 10]] ($88\%$)} & \multicolumn{3}{c|}{[[108, 8, 10]] ($88\%$)} & \multicolumn{3}{c|}{[[144, 12, 12]] ($90\%$)} & \multicolumn{3}{c|}{[[288, 12, 18]] ($88\%$)}\\
    \hline
     & Size & Deg & Time  & Size & Deg & Time  & Size & Deg & Time  & Size & Deg & Time  & Size & Deg & Time   \\ \hline
    CKBB  & 84/81 & 7/7 & 0.24 & 240/235 & 7/7 & 0.24 & 348/342 & 7/7 & 0.24 & 348/342 & 7/7 & 0.24 & 792/783 & 7/7 & 0.24 \\\hline
    Gauge & 90/91 & 8/8 & 0.05 & 157/158 & 8/10 & 0.06 & 185/186 & 8/9 & 0.07 & 239/240 & 8/10 & 0.07 & 362/363 & 8/11 & 0.09  \\\hline
    Exp-Opt  & 42/43 & 7/7 & 0.06 & 100/101 & 7/7 & 0.07 & 153/154 & 7/7 & 0.08 & 122/123 & 7/8 & 0.20 & 234/235 & 7/8 & 0.40 \\\hline
    Cong-Opt & 42/43 & 7/7 & 0.20 & 100/101 & 7/7 & 0.29 & 123/124 & 7/7 & 0.35 & 89/90 & 7/7 & 0.41 & 184/185 & 7/7 & 0.63 \\\hline
    Full-Opt & 9/10 & 7/7 & 0.21 & 19/20 & 7/7 & 0.57 & 23/24 & 7/7 & 0.81 & 24/25 & 7/8 & 0.87 & 45/46 & 9/9 & 1.80 \\ \hline
    \textbf{HP code} & \multicolumn{3}{c|}{[[450, 32, 8]] ($89\%$)} & \multicolumn{3}{c|}{[[882, 50, 10]] ($88\%$)} & \multicolumn{3}{c|}{[[882, 98, 8]] ($90\%$)} & \multicolumn{3}{c|}{[[1800, 72, 14]] ($88\%$)} & \multicolumn{3}{c|}{[[1922, 200, 12]] ($90\%$)}\\
    \hline
     & Size & Deg & Time  & Size & Deg & Time  & Size & Deg & Time  & Size & Deg & Time  & Size & Deg & Time   \\ \hline
    CKBB   & 152/148 & 7/7 & 0.24 & 240/235 & 7/7 & 0.24 & 225/217 & 9/9 & 0.24 & 600/589 & 9/9 & 0.24 & 432/419 & 11/11 & 0.25 \\\hline
    Gauge & 122/123 & 8/10 & 0.06 & 155/156 & 8/9 & 0.07 & 174/175 & 10/12 & 0.07 & 301/302 & 10/12 & 0.11 & 302/303 & 11/13 & 0.10 \\\hline
    Exp-Opt & 56/57 & 7/7 & 0.07 & 98/99 & 7/7 & 0.08 & 104/105 & 9/10 & 0.08 & 180/181 & 9/10 & 0.11 & 206/207 & 11/12 & 0.11 \\\hline
    Cong-Opt & 56/57 & 7/7 & 0.24 & 71/72 & 7/7 & 0.29 & 78/79 & 9/10 & 0.34 & 136/137 & 9/10 & 0.62 & 164/165 & 11/12 & 0.60 \\\hline
    Full-Opt & 14/15 & 7/7 & 0.41 & 19/20 & 7/7 & 0.63 & 18/19 & 9/10 & 0.33 & 35/36 & 9/10 & 1.15 & 30/31 & 11/12 & 0.59 \\ \hline
    \textbf{QT code} & \multicolumn{3}{c|}{[[81, 13, 4]] ($86\%$)} & \multicolumn{3}{c|}{[[90, 15, 4]] ($86\%$)} & \multicolumn{3}{c|}{[[117, 17, 4]] ($85\%$)} & \multicolumn{3}{c|}{[[126, 19, 4]] ($80\%$)} & \multicolumn{3}{c|}{[[135, 23, 4]] ($83\%$)}\\
    \hline
     & Size & Deg & Time  & Size & Deg & Time  & Size & Deg & Time  & Size & Deg & Time  & Size & Deg & Time   \\ \hline
    CKBB  & 28/28 & 7/7 & 0.24 & 28/28 & 7/7 & 0.24 & 78/76 & 7/7 & 0.23 & 220/217 & 7/7 & 0.24 & 312/309 & 7/7 & 0.24 \\\hline
    Gauge & 38/39 & 7/7 & 0.05 & 38/39 & 7/7 & 0.05 & 66/67 & 8/10 & 0.05 & 109/110 & 8/10 & 0.06 & 171/172 & 8/9 & 0.07 \\\hline
    Exp-Opt & 21/22 & 6/7 & 0.05 & 21/22 & 6/7 & 0.05 & 40/41 & 7/7 & 0.06 & 71/72 & 7/7 & 0.12 & 148/149 & 7/7 & 0.26 \\\hline
    Cong-Opt  & 5/6 & 6/7 & 0.14 & 5/6 & 6/7 & 0.14 & 40/41 & 7/7 & 0.18 & 73/74 & 7/7 & 0.29 & 123/124 & 7/7 & 0.40 \\\hline
    Full-Opt  & 4/5 & 6/7 & 0.14 & 4/5 & 6/7 & 0.14 & 10/11 & 7/7 & 0.30 & 22/23 & 8/7 & 0.71 & 29/30 & 8/7 & 0.97 \\ \hline
    \textbf{Surface code} & \multicolumn{3}{c|}{[[25, 1, 5]] ($87\%$)} & \multicolumn{3}{c|}{[[49, 1, 7]] ($87\%$)} & \multicolumn{3}{c|}{[[81, 1, 9]] ($86\%$)} & \multicolumn{3}{c|}{[[121, 1, 11]] ($84\%$)} & \multicolumn{3}{c|}{[[169, 1, 13]] ($87\%$)}\\
    \hline
     & Size & Deg & Time  & Size & Deg & Time  & Size & Deg & Time  & Size & Deg & Time  & Size & Deg & Time   \\ \hline
    CKBB  & 84/85 & 5/5 & 0.24 & 84/85 & 5/5 & 0.24 & 220/221 & 5/5 & 0.24 & 312/313 & 5/5 & 0.24 & 684/685 & 5/5 & 0.24  \\\hline
    Gauge & 102/103 & 8/9 & 0.06 & 102/103 & 8/9 & 0.06 & 165/166 & 8/9 & 0.07 & 200/201 & 8/10 & 0.07 & 381/382 & 8/10 & 0.09 \\\hline
    Exp-Opt & 39/40 & 5/6 & 0.10 & 39/40 & 5/6 & 0.10 & 76/77 & 5/7 & 0.36 & 90/91 & 5/7 & 0.48 & 248/249 & 5/7 & 0.94 \\\hline
    Cong-Opt & 47/48 & 5/7 & 0.28 & 47/48 & 5/7 & 0.28 & 76/77 & 5/7 & 0.47 & 124/125 & 5/7 & 0.56 & 243/244 & 5/7 & 0.89 \\\hline
    Full-Opt & 11/12 & 6/7 & 0.41 & 11/12 & 6/7 & 0.42 & 23/24 & 7/7 & 0.95 & 32/33 & 8/8 & 1.22 & 50/51 & 10/9 & 2.18 \\ \hline
    \textbf{Color code} & \multicolumn{3}{c|}{[[19, 1, 5]] ($83\%$)} & \multicolumn{3}{c|}{[[37, 1, 7]] ($87\%$)} & \multicolumn{3}{c|}{[[61, 1, 9]] ($87\%$)} & \multicolumn{3}{c|}{[[91, 1, 11]] ($86\%$)} & \multicolumn{3}{c|}{[[127, 1, 13]] ($88\%$)}\\
    \hline
     & Size & Deg & Time  & Size & Deg & Time  & Size & Deg & Time  & Size & Deg & Time  & Size & Deg & Time   \\ \hline
    CKBB  & 40/41 & 7/7 & 0.24 & 84/85 & 7/7 & 0.24 & 144/145 & 7/7 & 0.24 & 220/221 & 7/7 & 0.24 & 312/313 & 7/7 & 0.24 \\\hline
    Gauge & 53/54 & 8/8 & 0.05 & 102/103 & 8/9 & 0.06 & 137/138 & 8/10 & 0.06 & 164/165 & 8/9 & 0.07 & 254/255 & 8/10 & 0.07 \\\hline
    Exp-Opt & 27/28 & 6/7 & 0.09 & 39/40 & 6/7 & 0.10 & 59/60 & 6/7 & 0.26 & 101/102 & 6/7 & 0.35 & 91/92 & 6/7 & 0.47 \\\hline
    Cong-Opt & 20/21 & 6/7 & 0.20 & 47/48 & 6/7 & 0.28 & 64/65 & 6/7 & 0.38 & 76/77 & 6/7 & 0.47 & 124/125 & 6/7 & 0.56 \\\hline
    Full-Opt & 7/8 & 6/7 & 0.26 & 11/12 & 6/7 & 0.41 & 18/19 & 7/7 & 0.67 & 23/24 & 7/7 & 0.91 & 30/31 & 8/7 & 1.22 \\ \hline
    \end{tabular}%
    \end{adjustbox}
  }
  \vspace{-7pt}
\end{table}
We first measure the ancilla system generated for a single logical $X$ measurement; logical $Z$ measurements show similar behavior. This is the single-code building block used in the cross-code measurements of Section~\ref{Sec:Cross-CodeCommunication}. Table~\ref{tab:overall} uses the same convention as the later resource tables: \textbf{Size} is written as $a/b$, where $a$ is the number of introduced ancilla qubits and $b$ is the number of introduced checks; \textbf{Deg} is written as $a/b$, where $a$ is maximum qubit participation and $b$ is maximum check weight in the surgery system; and \textbf{Time} is wall-clock synthesis time in seconds. In Table~\ref{tab:overall}, the percentage shown next to each code instance is the reduction in ancilla qubits achieved by \textsf{Full-Opt} relative to the best baseline for that instance. We report raw qubit/check counts rather than only ratios because the absolute ancilla size determines whether a surgery protocol is practical at scale.

Taking the smaller of CKBB and Gauge as the baseline for each instance, \textsf{Full-Opt} reduces the number of ancillary qubits by $86.7\%$ on average and the number of ancillary checks by $85.8\%$ on average.
All single-code compilations complete within a few seconds, and the reported degrees remain within the bounded LDPC budget.

We further analyze the contribution of each optimization component. The conditioned expander construction (\textsf{Exp-Opt}) reduces qubit count by $41.0\%$ and check count by $40.3\%$ relative to the best baseline. Building on this, the congestion--expansion balancing technique (\textsf{Cong-Opt}) provides an additional reduction of $8.1\%$ for both qubits and checks. Finally, the degree-aware cellulation stage uses the remaining bounded LDPC budget to obtain a further reduction of $37.6\%$ in qubit count and $37.4\%$ in check count.

\textbf{Comparison with best-known results for specific code instances.}
We compare our constructions with IBM's ad-hoc solutions obtained via integer programming~\cite{williamson2024low} for the QEC codes used in their gross architecture.
For BB codes $[[144,12,12]]$ and $[[288,12,18]]$, the best-known ancillary sizes by IBM are $22/19$ and $34/31$ (qubits/checks), with qubit and check degrees both equal to $7$. In comparison, \myCompilerNameSpace produces $24/25$ and $45/46$, with degrees $7/8$ and $9/9$, respectively. Thus, our general synthesis framework remains close to instance-specific optimized solutions while retaining a by-construction correctness path. Note that IBM's ad-hoc solution also requires integer-programming-based checking for distance preservation.

\begin{table}[t]
  \centering
  \caption{Variability of \myCompilerNameSpace over its 100 runs and comparison with SSIP (Size: ancillary qubits/checks; Deg: qubit/check degree; time in seconds)}
  \vspace{-10pt}
  \label{tab:variability}%
  {\fontsize{8}{10}\selectfont
    \setlength{\tabcolsep}{2pt}%
    \renewcommand{\arraystretch}{0.95}%
    
    \begin{adjustbox}{max width=\textwidth}
    \begin{tabular}{|c|c|c|c|c|c|c|c|c|c|c|c|c|}
    \hline
    \multicolumn{2}{|c|}{\multirow{2}{*}{Code}} & Best baseline & \multicolumn{3}{c|}{SSIP} & \multicolumn{3}{c|}{Full-Opt (Table~\ref{tab:overall})} & \multicolumn{2}{c|}{Full-Opt, 100 runs: qubits} & \multicolumn{2}{c|}{Full-Opt, 100 runs: checks} \\ \cline{3-13}
    \multicolumn{2}{|c|}{} & Size & Size & Deg & Time & Size & Deg & Time & Min, median, max & Mean$\pm$std & Min, median, max & Mean$\pm$std \\ \hline
    \multirow{5}{*}{BB} & [[72, 12, 6]] & 84/81 & 9/6 & 7/7 & 8.7 & 9/10 & 7/7 & 0.2 & 9, 9, 10 & 9.4$\pm$0.5 & 10, 10, 11 & 10.4$\pm$0.5 \\ \cline{2-13}
     & [[90, 8, 10]] & 157/158 & 15/10 & 7/7 & 8.3 & 19/20 & 7/7 & 0.6 & 19, 20, 24 & 20.5$\pm$1.2 & 20, 21, 25 & 21.5$\pm$1.2 \\ \cline{2-13}
     & [[108, 8, 10]] & 185/186 & 18/12 & 7/7 & 9.1 & 23/24 & 7/7 & 0.8 & 23, 26, 30 & 26.4$\pm$1.6 & 24, 27, 31 & 27.4$\pm$1.6 \\ \cline{2-13}
     & [[144, 12, 12]] & 239/240 & 78/72 & 7/7 & 53.7 & 24/25 & 7/8 & 0.9 & 24, 28, 68 & 28.2$\pm$4.4 & 25, 29, 69 & 29.2$\pm$4.4 \\ \cline{2-13}
     & [[288, 12, 18]] & 362/363 & 117/108 & 7/7 & 166 & 45/46 & 9/9 & 1.8 & 43, 49, 123 & 63.3$\pm$26.9 & 44, 50, 124 & 64.3$\pm$26.9 \\ \hline
    \multirow{5}{*}{HP} & [[450, 32, 8]] & 122/123 & 12/8 & 7/7 & 310 & 14/15 & 7/7 & 0.4 & 14, 16, 17 & 15.5$\pm$0.7 & 15, 17, 18 & 16.5$\pm$0.7 \\ \cline{2-13}
     & [[882, 50, 10]] & 155/156 & 15/10 & 7/7 & 767 & 19/20 & 7/7 & 0.6 & 19, 21, 25 & 21.3$\pm$1.1 & 20, 22, 26 & 22.3$\pm$1.1 \\ \cline{2-13}
     & [[882, 98, 8]] & 174/175 & 17/9 & 9/9 & 1796 & 18/19 & 9/10 & 0.3 & 18, 18, 18 & 18.0$\pm$0.0 & 19, 19, 19 & 19.0$\pm$0.0 \\ \cline{2-13}
     & [[1800, 72, 14]] & 301/302 & 26/15 & 9/9 & 1814 & 35/36 & 9/10 & 1.1 & 35, 39, 106 & 41.1$\pm$12.1 & 36, 40, 107 & 42.1$\pm$12.1 \\ \cline{2-13}
     & [[1922, 200, 12]] & 302/303 & 25/12 & 11/11 & 10601 & 30/31 & 11/12 & 0.6 & 30, 30, 30 & 30.0$\pm$0.0 & 31, 31, 31 & 31.0$\pm$0.0 \\ \hline
    \multirow{5}{*}{QT} & [[81, 13, 4]] & 28/28 & 4/4 & 6/7 & 11.5 & 4/5 & 6/7 & 0.1 & 4, 4.5, 6 & 4.7$\pm$0.8 & 5, 5.5, 7 & 5.7$\pm$0.8 \\ \cline{2-13}
     & [[90, 15, 4]] & 28/28 & 4/4 & 6/7 & 13.9 & 4/5 & 6/7 & 0.1 & 4, 5, 6 & 4.7$\pm$0.8 & 5, 6, 7 & 5.7$\pm$0.8 \\ \cline{2-13}
     & [[117, 17, 4]] & 66/67 & 8/6 & 7/7 & 28.7 & 10/11 & 7/7 & 0.3 & 10, 11, 14 & 11.2$\pm$0.8 & 11, 12, 15 & 12.2$\pm$0.8 \\ \cline{2-13}
     & [[126, 19, 4]] & 109/110 & 13/10 & 7/7 & 28.2 & 22/23 & 8/7 & 0.7 & 21, 24, 28 & 23.9$\pm$1.7 & 22, 25, 29 & 24.9$\pm$1.7 \\ \cline{2-13}
     & [[135, 23, 4]] & 171/172 & 42/39 & 7/7 & 110 & 29/30 & 8/7 & 1.0 & 26, 31, 38 & 30.7$\pm$2.4 & 27, 32, 39 & 31.7$\pm$2.4 \\ \hline
    \end{tabular}%
    \end{adjustbox}
  }
  \vspace{-7pt}
\end{table}

\textbf{Run-to-run variability.} Table~\ref{tab:variability} shows the distribution of the 100 independent runs behind Table~\ref{tab:overall} for the BB, HP, and QT codes, with the port pairings fixed. The median run uses 0--25\% more ancillary qubits and 0--20\% more checks than the reported construction; the largest outliers are runs in which an unfavorable sample needs a second thickening layer. Even the worst of the 100 runs reduces ancillary qubits and checks by 79.8\% and 78.8\% on average relative to the best baseline, and the median run by 85.8\% and 84.8\%.

\textbf{Comparison with SSIP.} We run SSIP~\cite{cowtan2024ssip} on the same codes and logical operators. SSIP adds layers to its ancilla system until a heuristic estimate of an upper bound on the distance (QDistRnd~\cite{pryadko2023qdistrnd}) reaches $d$, so it does not guarantee distance preservation. Its time, which includes this estimation, reaches about three hours for HP $[[1922,200,12]]$, whereas \myCompilerNameSpace completes all 100 runs within 2\,s. SSIP also leaves the gauge degrees of freedom of its ancilla system unfixed, so its check counts are not directly comparable. \myCompilerNameSpace uses slightly more ancillary qubits than SSIP on most codes, but fewer on the larger BB and QT codes, where SSIP needs 3.3$\times$, 2.6$\times$, and 1.4$\times$ the ancillary qubits of \myCompilerNameSpace on BB $[[144,12,12]]$, BB $[[288,12,18]]$, and QT $[[135,23,4]]$ (Table~\ref{tab:variability}).

\subsection{Cross-Code Communication} \label{Sec:Cross-CodeCommunication}
We next evaluate the ancilla overhead of measuring joint logical operators between codes. We report results for $X \otimes X$ operators, while other Pauli combinations show similar behavior. We select codes with similar distances from Table~\ref{tab:overall} and connect them using the bridge/adapter constructions~\cite{swaroop2024universal}.

\begin{table}[!htbp]
  \centering
  \caption{Resources for communication between codes with close distances}
  \vspace{-10pt}
  {\fontsize{8}{10}\selectfont
    \setlength{\tabcolsep}{2pt}%
    \renewcommand{\arraystretch}{0.95}%
    \begin{adjustbox}{max width=\textwidth}
    \begin{tabular}{|m{1.6cm}<{\centering}|c|c|c|c|c|c|c|c|c|c|c|c|c|c|c|}
    \hline
    \textbf{Code config} & \multicolumn{3}{c|}{\begin{tabular}{c} BB [[72, 12, 6]] +\\ Surface [[25, 1, 5]] \end{tabular}} & \multicolumn{3}{c|}{\begin{tabular}{c} BB [[90, 8, 10]] +\\ Surface [[81, 1, 9]] \end{tabular}} & \multicolumn{3}{c|}{\begin{tabular}{c} BB [[144, 12, 12]] +\\ Surface [[121, 1, 11]] \end{tabular}} & \multicolumn{3}{c|}{\begin{tabular}{c} HP [[1922, 200, 12]] +\\ Surface [[121, 1, 11]] \end{tabular}} & \multicolumn{3}{c|}{\begin{tabular}{c} QT [[135, 23, 4]] +\\ Surface [[25, 1, 5]] \end{tabular}}\\
    \hline
     & Size & Deg & Time  & Size & Deg & Time  & Size & Deg & Time  & Size & Deg & Time  & Size & Deg & Time   \\ \hline
    Gauge & 197/199 & 8/9 & 0.11 & 331/333 & 8/11 & 0.14 & 450/452 & 8/11 & 0.15 & 513/515 & 11/14 & 0.18 & 277/279 & 8/10 & 0.13 \\ \hline
    Full-Opt & 25/27 & 7/7 & 0.62 & 51/53 & 9/9 & 1.51 & 67/69 & 10/9 & 2.07 & 73/75 & 11/12 & 1.82 & 44/46 & 9/8 & 1.38 \\ \hline
    \textbf{Code config} & \multicolumn{3}{c|}{\begin{tabular}{c} BB [[72, 12, 6]] +\\ Color [[19, 1, 5]] \end{tabular}} & \multicolumn{3}{c|}{\begin{tabular}{c} BB [[90, 8, 10]] +\\ Color [[61, 1, 9]] \end{tabular}} & \multicolumn{3}{c|}{\begin{tabular}{c} BB [[144, 12, 12]] +\\ Color [[91, 1, 11]] \end{tabular}} & \multicolumn{3}{c|}{\begin{tabular}{c} HP [[1922, 200, 12]] +\\ Color [[91, 1, 11]] \end{tabular}} & \multicolumn{3}{c|}{\begin{tabular}{c} QT [[135, 23, 4]] +\\ Color [[19, 1, 5]] \end{tabular}}\\
    \hline
     & Size & Deg & Time  & Size & Deg & Time  & Size & Deg & Time  & Size & Deg & Time  & Size & Deg & Time   \\ \hline
    Gauge & 148/150 & 8/9 & 0.11 & 303/305 & 8/11 & 0.13 & 414/416 & 8/11 & 0.14 & 477/479 & 11/14 & 0.18 & 228/230 & 8/10 & 0.12 \\\hline
    Full-Opt & 21/23 & 8/7 & 0.47 & 45/47 & 9/8 & 1.22 & 58/60 & 9/8 & 1.78 & 64/66 & 11/12 & 1.51 & 40/42 & 9/8 & 1.22 \\ \hline
    \textbf{Code config} & \multicolumn{3}{c|}{\begin{tabular}{c} BB [[144, 12, 12]] +\\ BB [[144, 12, 12]] \end{tabular}} & \multicolumn{3}{c|}{\begin{tabular}{c} HP [[1922, 200, 12]] +\\ HP [[1922, 200, 12]] \end{tabular}} & \multicolumn{3}{c|}{\begin{tabular}{c} QT [[135, 23, 4]] +\\ QT [[135, 23, 4]] \end{tabular}} & \multicolumn{3}{c|}{\begin{tabular}{c} BB [[90, 8, 10]] +\\ HP [[882, 50, 10]] \end{tabular}} & \multicolumn{3}{c|}{\begin{tabular}{c} BB [[144, 12, 12]] +\\ HP [[1922, 200, 12]] \end{tabular}}\\
    \hline
     & Size & Deg & Time  & Size & Deg & Time  & Size & Deg & Time  & Size & Deg & Time  & Size & Deg & Time   \\ \hline
    Gauge & 490/492 & 8/11 & 0.15 & 616/618 & 11/14 & 0.21 & 346/348 & 8/10 & 0.14 & 322/324 & 8/11 & 0.14 & 553/555 & 11/14 & 0.18  \\ \hline
    Full-Opt & 60/62 & 9/8 & 1.74 & 72/74 & 11/12 & 1.18 & 58/60 & 9/10 & 1.94 & 48/50 & 9/8 & 1.17 & 66/68 & 11/12 & 1.48 \\ \hline
    \end{tabular}%
    \end{adjustbox}
  }
  \label{tab:qldpc_topologic}%
  \vspace{-6pt}
\end{table}

First, we evaluate the code surgery system overhead of measuring joint logical operators between QLDPC codes and topological codes with comparable distances (rows 1 and 2 in Table~\ref{tab:qldpc_topologic}). \myCompilerNameSpace can accommodate codes with very different distances, but such a configuration is unlikely to be useful in practice because the overall fault tolerance is determined by the lower distance of the two codes being connected.
We consider BB, HP, and QT codes combined with surface and color codes. Compared with the baseline (Gauge), our optimized approach (\textsf{Full-Opt}) achieves a reduction of approximately $85\%$ in both ancillary qubit and check counts across all configurations, demonstrating its effectiveness in minimizing the resource cost of cross-family logical operations.

We further evaluate the communication overhead between QLDPC codes (row group 3 in Table~\ref{tab:qldpc_topologic}), considering both homogeneous (e.g., BB+BB, HP+HP, QT+QT) and heterogeneous (e.g., BB+HP) configurations. Cross-family settings involving QT codes are omitted due to their relatively small distance in our benchmark. Compared with the baseline (Gauge), our optimized construction (\textsf{Full-Opt}) achieves a reduction of approximately $86\%$ in both ancillary qubit and check counts across all configurations, demonstrating its effectiveness in enabling resource-efficient logical communication within and across QLDPC codes.

\subsection{Logical Error Rate Simulation} \label{Sec:LogicalErrorRateSimulation}
\myCompilerName's output is a code structure independent of a particular decoder but we are still interested in the performance of our synthesized code structure in a complete syndrome-measurement-decoding pipeline.
For this evaluation, we simulate logical error rates under a circuit-level depolarizing noise model: each single- or two-qubit gate is followed by an independent depolarizing channel, and each measurement has an independent measurement error. We obtain the syndrome-measurement circuits for both methods using the qLDPC software package~\cite{perlin2023qldpc}. We decode both the baseline (Gauge) and optimized \myCompilerNameSpace constructions using the same BP-OSD decoder~\cite{roffe2020decoding,fossorier1995soft}, i.e., belief propagation followed by ordered-statistics decoding post-processing. As shown in Figure~\ref{fig:simulation}, each data point is estimated from $10^6$ samples; missing points indicate that no logical errors were observed within $10^6$ samples. Across all tested configurations, including standalone BB codes and their combinations with surface and color codes, the optimized \myCompilerNameSpace constructions (using \textsf{Full-Opt}) consistently achieve logical error rates that closely match those of the baseline (Gauge).

\begin{figure}[!htbp]
    \centering
    \vspace{-2pt}
    \includegraphics[width=1.0\linewidth]{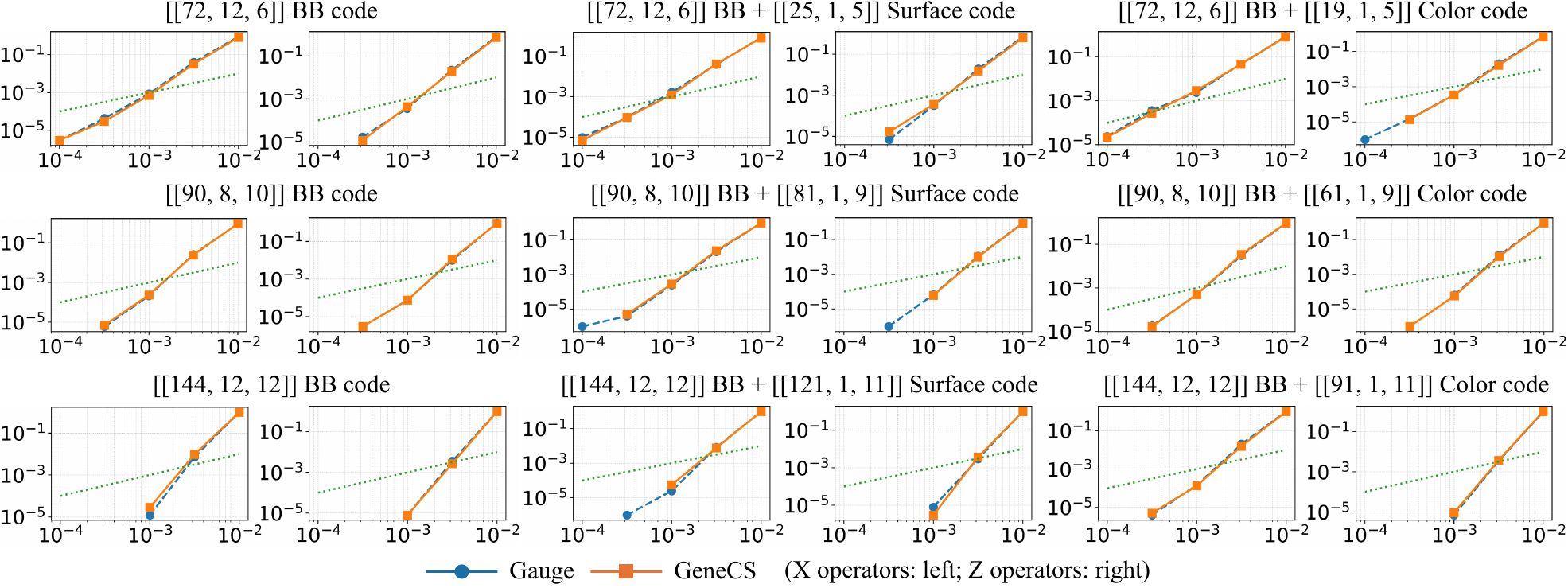}
    \vspace{-20pt}
    \caption{Logical error rate simulation under a circuit-level noise model. Green dotted lines indicate physical error rate equal to logical error rate.}
    \label{fig:simulation}
    \vspace{-8pt}
\end{figure}

Notably, this comparable behavior is maintained despite a substantial reduction in ancillary resources (Tables~\ref{tab:overall} and~\ref{tab:qldpc_topologic}). The logical error curves of the optimized codes nearly overlap with those of the baseline across a wide range of physical error rates, indicating that the resource savings do not introduce an observable loss in error-correcting performance in this simulation.
The savings carry over to the generated circuits. Per syndrome-extraction round and averaged over the nine configurations, \myCompilerNameSpace reduces the two-qubit gates and fault locations (noise-channel locations of the noise model) relative to Gauge by 51.3\% and 52.3\% for the whole circuit, and by 84.6\% and 85.1\% for the part added by surgery, i.e., excluding the measurement of the checks of the original codes. The logical error rates are comparable rather than lower because the original codes dominate: they account for 73--90\% of the two-qubit gates of the \myCompilerNameSpace circuits, and since the ancilla system adds no logical qubits, the dominant low-weight logical failures are those of the original codes. Moreover, a logical operator dressed through a larger expanding ancilla system has higher weight, so the extra fault locations of a larger construction need not raise the logical error rate proportionally.
Overall, these results demonstrate that \myCompilerNameSpace improves resource efficiency while preserving comparable logical-error-rate behavior under the tested decoder and noise model.

\subsection{Scalability} \label{Sec:Scalability}
We evaluate the scalability of our compiler on the surface code family with distances up to $d=101$, as shown in Table~\ref{tab:scale}. This stress test scales to codes with more than $10^4$ physical qubits while keeping maximum qubit participation and maximum check weight at $12$. The algorithmic cost of \myCompilerNameSpace is governed by the synthesized ancilla graph and the support of the target logical operator, not by the total number of physical qubits alone.
The total synthesis runtime grows moderately with the code size. Table~\ref{tab:scale} reports the total time for a batch of 100 independent, valid synthesis runs, reaching $103.81$ seconds for the largest input. Thus one synthesized surgery graph takes $1.04$ seconds on average for that input. 


\begin{table}[!htbp]
  \centering
  \caption{Surface-code scalability; time is total for 100 independent synthesis runs}
  \vspace{-10pt}
  {\fontsize{8}{10}\selectfont
    \setlength{\tabcolsep}{2pt}%
    \renewcommand{\arraystretch}{0.95}%
    \begin{adjustbox}{max width=\textwidth}
    \begin{tabular}{|c|c|c|c|c|c|c|c|c|c|c|c|c|c|c|}
    \hline
    \multicolumn{3}{|c|}{[[441, 1, 21]]} & \multicolumn{3}{c|}{[[1681, 1, 41]]} & \multicolumn{3}{c|}{[[3721, 1, 61]]} & \multicolumn{3}{c|}{[[6561, 1, 81]]} & \multicolumn{3}{c|}{[[10201, 1, 101]]}\\
    \hline
     Size & Deg & Time  & Size & Deg & Time  & Size & Deg & Time  & Size & Deg & Time  & Size & Deg & Time   \\ \hline
    149/150 & 12/12 & 7.26 & 360/361 & 12/12 & 23.04 & 777/778 & 12/12 & 42.36 & 1086/1087 & 12/12 & 69.51 & 1708/1709 & 12/12 & 103.81  \\ \hline
    \end{tabular}%
    \end{adjustbox}
  }
  \label{tab:scale}%
  \vspace{-6pt}
\end{table}

To further understand scalability, we analyze the theoretical complexity of our construction. The dominant cost arises from repeatedly estimating the second eigenvalue $\lambda_2(G)$ during graph construction. Here $n$ denotes the number of vertices in the synthesized ancilla graph, which is tied to the logical-operator support size, not necessarily the total physical-qubit count of the input code. In the worst case, each dense eigenvalue estimation requires $O(n^3)$ time, and since we perform $O(n)$ such updates due to constant-degree expansion, the overall worst-case complexity is $O(n^4)$ in the ancilla-graph size. In practice, $\lambda_2$ can also be approximated or updated incrementally; our current implementation recomputes it directly to keep the prototype simple.

\begin{remark}
Although our scalability evaluation focuses on surface codes, the trend is expected to extend more broadly to other stabilizer codes. In particular, prior work has shown that irreducible logical operators in CSS stabilizer codes admit a locally repetitive structure analogous to that of the surface code~\cite{cowtan2024css}. Consequently, we expect the runtime and scaling behavior observed here to be broadly representative beyond the surface-code setting.
\end{remark}

%% file: tex/8-RelatedWork.tex
\section{Related Work} \label{Sec:RelatedWork}

Prior work on code surgery for general stabilizer codes follows two complementary directions: analytic constructions that obtain correctness from structural conditions, and automated methods that optimize individual instances. Among analytic constructions, CKBB~\cite{cohen2022low} preserves distance but introduces low-distance gauge degrees of freedom that require thickening. Gauge-fixing-based approaches~\cite{williamson2024low, cross2024improved} provide QLDPC guarantees and support extensions such as extractors~\cite{he2025extractors} and adapters~\cite{swaroop2024universal}. Branching constructions build first on CKBB~\cite{zhang2025time} and later on gauge fixing~\cite{cowtan2025parallel}. These methods provide structural guarantees, but their conservative constructions can incur substantial ancilla overhead.

Automated and numerical tools such as SSIP~\cite{cowtan2024ssip,cowtan2024css}, high-rate surgery~\cite{zheng2025high}, and optimization-based protocols~\cite{ide2025fault} can reduce overhead on concrete instances. These methods take a different route to correctness: distance preservation is typically established through post-compilation verification or expensive optimization. In particular, SSIP can produce surgery instances, but it does not provide a distance-preservation guarantee by construction, so the generated code must be checked after synthesis; Section~\ref{Sec:OverallResult} quantifies this trade-off. This verification entails code-distance search, which is NP-hard already for classical linear codes~\cite{vardy1997algorithmic} and is expensive in practice for quantum codes~\cite{pryadko2023qdistrnd}. Moreover, verifying constructions individually does not by itself establish distance preservation after they are composed.

\myCompilerNameSpace combines these directions. It retains the by-construction distance and QLDPC guarantees of graph-based surgery while exposing the surgery graph as a compiler IR subject to resource optimization. The IR makes the generated stabilizer code, correctness conditions, and resource cost model explicit, enabling compiler optimizations to scale without post-compilation distance verification.

%% file: tex/9-Statements.tex
\section*{Data Availability Statement}
\myCompilerNameSpace is available on GitHub at \url{https://github.com/HQEC/GeneCS}. In addition, an archived version is provided on Zenodo~\cite{genecs2026artifact}. Both of them include detailed instructions and scripts that facilitate the reproducibility of the experiments presented in the paper.

\section*{Disclosure of Use of AI in Research}
The design and algorithms of \myCompilerNameSpace were developed by the authors. Claude Code (Anthropic) was used to help debug the \myCompilerNameSpace implementation, and to write the experiment scripts and collect the data for Table~\ref{tab:variability} and the circuit-level gate and fault-location counts in Section~\ref{Sec:LogicalErrorRateSimulation}. The authors reviewed all code changes and scripts produced with its assistance.

\section*{Acknowledgements}
We thank the anonymous reviewers for their insightful feedback. We thank Zhiyang He (Sunny), Alexander Cowtan,  Dominic J. Williamson, and Theodore J. Yoder for their insightful discussions and valuable feedback. This work was supported in part by the U.S. Department of Energy, Office of Science, Office of Advanced Scientific Computing Research through the Accelerated Research in Quantum Computing Program MACH-Q project, the U.S. Department of Energy Early Career Research Award, the U.S. National Science Foundation CAREER Award No. CCF-2338773, and the Defense Advanced Research Projects Agency
Microsystems Technology Office (MTO) HARQ Program under Agreement No. HR0011-26-9-E113. GL was also supported by a Cisco Research Award. This research was, in part, funded by the U.S. Government.
The views and conclusions contained in this document are those of the authors and should not be interpreted as representing the official policies, either expressed or implied, of the U.S. Government.

%% file: tex/Appendix.tex
\section{Background on Graph-Based Code Surgery Construction}
\label{app:graph_construction}

This section provides the formal background on graph-based code surgery constructions. This material is included for completeness and is not required to understand the main contributions of the paper.

\subsection{Ancilla System for Code Surgery} \label{Sec:AncillaSystemasaGraphapp}

\begin{figure*}[t]
    \centering
    \includegraphics[width=0.9\linewidth]{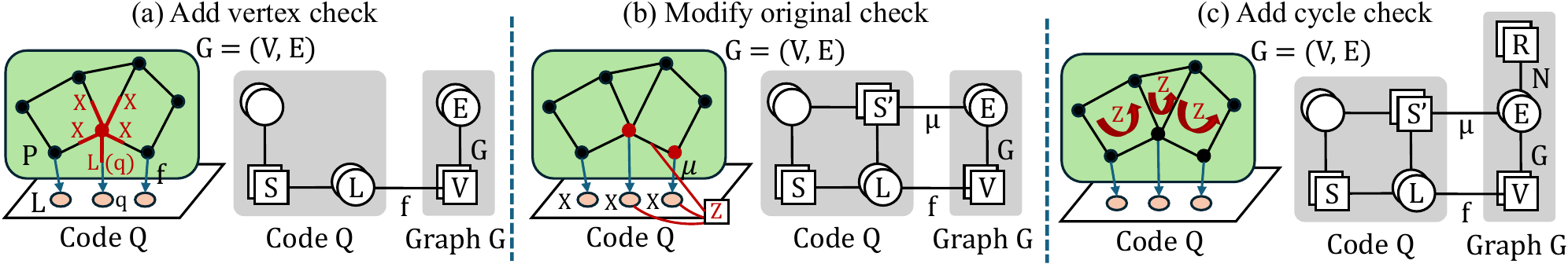}
    \vspace{-10pt}
    \caption{Ancilla system represented as a graph. We use circle to represent qubits and squares to represent checks following the convention in QEC literature.}
    \label{fig:Ancillasystemasgraph_appendix}
\end{figure*}

Lattice surgery~\cite{horsman2012surface, litinski2019game} is a technique for performing fault-tolerant logical operations on surface codes through measurements of joint operators between encoded qubits.

Recently, lattice surgery has been extended to generalized code surgery~\cite{he2025extractors, williamson2024low, cohen2022low, cross2024improved, swaroop2024universal}, enabling a broader class of fault-tolerant logical operations across different code families. This is achieved by deforming the original stabilizer code into a new code in which the target logical operator is incorporated as a stabilizer.

We refer to the additional qubits and checks introduced during code surgery as an ancilla system, which can be represented as a graph. This process is shown in Figure~\ref{fig:Ancillasystemasgraph_appendix}.

We denote the original stabilizer code by $\mathcal{Q}$ with stabilizer group $\mathcal{S}$. Let $L$ be the logical operator of interest, represented as a Pauli string supported on a set of physical qubits $q_i$. We write $L(q_i)$ to denote the Pauli operator acting on qubit $q_i$. Then the ancilla system is represented as a graph $G = (V, E)$ attached to the original code $\mathcal{Q}$, defined via a bijective map $f : \operatorname{supp}(L) \to P$, which connects the support of the logical operator to a subset of vertices called ports. This construction is illustrated in Fig.~\ref{fig:Ancillasystemasgraph_appendix}(a).

We now formally define the code surgery construction by interpreting the graph $G = (V, E)$ together with the original code $\mathcal{Q}$ as a deformed stabilizer code $\widetilde{\mathcal{Q}}$.

\textit{\textbf{First}}, we introduce vertex checks. We use the $X/Z$ orientation for the CSS ancilla system; the dual orientation exchanges the roles of $X$ and $Z$. The base code need not be CSS: the construction uses the actual local Pauli factors of its checks and logical operator. Each vertex $v \in V$ corresponds to a check, and each edge $e \in E$ corresponds to a qubit. For $v \in V$, if $v \notin P$, we add the stabilizer
$A_v = \prod_{e \ni v} X(e)$
to the deformed stabilizer group $\widetilde{\mathcal{S}}$; if $v = f(q) \in P$, we instead add
$A_v = L(q)\prod_{e \ni v} X(e)$
to $\widetilde{\mathcal{S}}$. An example is illustrated in Fig.~\ref{fig:Ancillasystemasgraph_appendix}, where the vertex highlighted in red enforces an $X$-parity check on its adjacent edges ($e \ni v$), together with the Pauli operator on the logical support of $L$ specified by the map $f$. Since we consider connected graphs, the product of all vertex checks yields the logical operator to be measured. The high-level representation on the right of Fig.~\ref{fig:Ancillasystemasgraph_appendix}(a) uses double squares for checks and double circles for qubits, with incidence determined by the graph $G$ and the port map $f$.

\textit{\textbf{Second}}, we modify the original stabilizers to ensure that the deformed stabilizer group remains commuting. This is necessary because the vertex checks introduced in the previous step may anti-commute with some stabilizers in the original code. The procedure is as follows. For each stabilizer $S \in \mathcal{S}$, let $K(S, L)$ denote the set of qubits $q$ on which $S(q)$ and $L(q)$ anti-commute. Note that $|K(S, L)|$ must be even. If $K(S, L) = \emptyset$, we directly include $S$ in $\widetilde{\mathcal{S}}$. Otherwise, let $\mu(S, L)$ be a path matching of $f(K(S, L))$ in $G$, and define the modified stabilizer $\bar{S} = S \prod_{e \in \mu(S, L)} Z(e)$, which is then added to $\widetilde{\mathcal{S}}$. An example is shown in Fig.~\ref{fig:Ancillasystemasgraph_appendix}(b), where the logical operator is $L = XXX$, and a stabilizer $S = ZZ$ anti-commutes with two $X$ operators in $L$'s physical support on the right. These two physical qubits form the set $K(S, L)$, and $f(K(S, L))$ corresponds to the two highlighted red vertices in $G$. A valid path matching $\mu(S, L)$ can be chosen as the edge connecting these two vertices. Consequently, the modified stabilizer is obtained by augmenting $S$ with $Z$ operators acting on the edge qubit along this path matching. The high-level representation is shown on the right of Fig.~\ref{fig:Ancillasystemasgraph_appendix}(b), where the modified stabilizer $S'$ is connected to the edge qubits $E$ via the path matching $\mu$.

\textit{\textbf{Third}}, cycles in the graph $G$ correspond to gauge degrees of freedom introduced during code surgery. The associated logical operators can dress logical operators of the original code, thereby potentially reducing the code distance. Therefore, it is necessary to fix this gauge freedom. Let $R$ be a cycle basis of $G$. For each cycle $C \in R$, we add the stabilizer $B_C = \prod_{e \in C} Z(e)$ to $\widetilde{\mathcal{S}}$. An example and its high-level representation are shown in Fig.~\ref{fig:Ancillasystemasgraph_appendix}(c), where three cycle checks are added to the deformed code.

\subsection{Graph Construction} \label{Sec:graphconstructionapp}

This construction imposes several structural requirements on the ancilla graph. We require that it preserves the code distance, and that the deformed code remains a QLDPC code if the original code belongs to the QLDPC family. Prior work~\cite{williamson2024low, he2025extractors} provides a theoretical procedure for constructing graphs that satisfy these conditions, which we briefly introduce here. We first present several definitions.

\begin{definition}[Cheeger Constant and Relative Cheeger Constant~\cite{he2025extractors}] \label{Def:Cheeger}
Let $G = (V, E)$ be an undirected graph. For any subset $U \subseteq V$, the edge boundary of $U$, denoted by $\delta(U)$, is defined as the set of edges with exactly one endpoint in $U$.

The Cheeger constant $\beta(G)$ is defined as the largest real number such that, for all $U \subseteq V$,
\begin{equation*}
|\delta(U)| \ge \beta(G)\cdot \min(|U|, |V \setminus U|).
\end{equation*}

Furthermore, for a subset $P \subseteq V$ and an integer $t > 0$, the relative Cheeger constant $\beta_t(G, P)$ is defined as the largest real number such that, for all $U \subseteq V$,
\begin{equation*}
|\delta(U)| \ge \beta_t(G, P)\cdot \min\big(t, |U \cap P|, |P \setminus U|\big).
\end{equation*}

In particular, the standard Cheeger constant is recovered as $\beta(G) = \beta_{|V|}(G, V)$. Both quantities are commonly referred to as measures of graph expansion, with larger values indicating better expansion.
\end{definition}

\begin{definition}[Congestion~\cite{he2025extractors}]
Let $G = (V, E)$ be a graph, and let $R$ be a cycle basis of $G$. For each edge $e \in E$, let $\rho_R(e)$ denote the number of cycles in $R$ that contain $e$. Define
\[
\rho = \max_{e \in E} \rho_R(e).
\]
We say that $R$ has congestion $\rho$, or that $R$ is a $\rho$-basis.
\end{definition}

\begin{definition}[Graph Thickening~\cite{he2025extractors}]
Let $G_1 = (V_1, E_1)$ and $G_2 = (V_2, E_2)$ be two graphs. The Cartesian product of $G_1$ and $G_2$ is the graph 
\[
G = G_1 \square G_2 = (V_1 \times V_2, E),
\]
where
\[
E = \{ ((u_1, u_2), (v_1, v_2)) : (u_2=v_2 \text{ and } (u_1,v_1)\in E_1) \text{ or } (u_1=v_1 \text{ and } (u_2,v_2)\in E_2) \}.
\]

Let $J_\ell$ denote a path graph of length $\ell$, consisting of $\ell$ vertices and $\ell - 1$ edges. We define the $\ell$-fold thickening of $G$ as the graph $G \square J_\ell$.

We refer to the $\ell$ copies of $G$ in $G \square J_\ell$ as \emph{levels}, and denote the $r$-th level by $G \times \{r\} = (V \times \{r\}, E \times \{r\})$ for $1 \le r \le \ell$.
\end{definition}

For a logical operator $L$, we construct a graph $G$ and a port function $f$ as shown in Algorithm~\ref{alg:AncillaGraphConstruction}, introduced in~\cite{williamson2024low, he2025extractors}.

    
        
    
    
    

\begin{algorithm}[t]
\caption{Ancilla Graph Construction}
\label{alg:AncillaGraphConstruction}
    \textbf{Step 1: Port connection}. Let $V_1$ be a set of vertices with $|V_1| = |\operatorname{supp}(L)|$, and construct a bijection $f$ between $\operatorname{supp}(L)$ and $V_1$.
    
    \textbf{Step 2: Expanding}. Construct a base graph $G_1 = (V_1, E_1)$ as follows: (a). For each stabilizer $S \in \mathcal{S}$, recall that $K(S, L)$ is the set of qubits $q \in \mathcal{Q}$ such that $S(q)$ and $L(q)$ anti-commute. Add a perfect matching $\mu(S, L)$ on $f(K(S, L))$, consisting of $|K(S, L)|/2$ edges, to $G_1$. (b). Construct a constant-degree graph $D$ on $|\operatorname{supp}(L)|$ vertices with Cheeger constant $\beta_D \ge \beta$ for some constant $\beta$. Add the edges of $D$ to $G_1$.

    \textbf{Step 3: Decongestion}. Apply the Decongestion Lemma (see Lemma A.0.2 in~\cite{freedman2021building}) to obtain a cycle basis $\mathcal{R}$ of $G_1$ with congestion $\rho$, together with a partition (see Corollary~15 in~\cite{he2025extractors}) $\mathcal{R} = \bigcup_{i=1}^{t} P_i$ such that each $P_i$ consists of edge-disjoint cycles and $t \leq O((\log |\operatorname{supp}(L)|)^3)$.
    
    \textbf{Step 4: Thickening}. Thicken $G_1$ by a factor $\ell = \max(t, 1/\beta)$ to obtain $G = G_1 \square J_\ell$, where $G = (V, E)$. This operation amplifies the relative Cheeger constant from $\beta$ to $\ell \beta$, and spreads the cycle basis across different levels $G_1 \times \{r\}$.
    
    \textbf{Step 5: Cellulation}. For each level $G_1 \times \{r\}$, triangulate every cycle in $P_r$ to obtain a collection of triangles that generate the cycles in $P_r \times \{r\}$.
\end{algorithm}

The purpose of Steps 2(b) and 4 is to ensure that logical operators in the original code $\mathcal{Q}$, when dressed by the newly introduced stabilizers, do not suffer a reduction in weight, thereby preserving the code distance. Steps 3 and 4 ensure that each qubit participates in only a constant number of cycle checks, while Step 5 guarantees that each cycle check has constant weight.

An example is shown in Figure~\ref{fig:GraphConstruction}. Suppose the base graph has Cheeger constant $0.5$ and admits three cycles in its cycle basis. By thickening the graph with $\ell = 2$, we obtain a graph with relative Cheeger constant $2 \cdot 0.5 = 1$, which meets the threshold required to preserve the code distance. Moreover, the cycle basis is distributed across levels such that cycles within each level are edge-disjoint, ensuring that each qubit participates in only a constant number of cycle checks (as illustrated in blue and yellow cycles in Figure~\ref{fig:GraphConstruction}(b)). Finally, the cellulation step reduces the cycle length (i.e., the weight of cycle checks) to 3 by decomposing cycles into triangles, as shown in Figure~\ref{fig:GraphConstruction}(c).

\begin{figure}[t]
    \centering
    \includegraphics[width=0.9\linewidth]{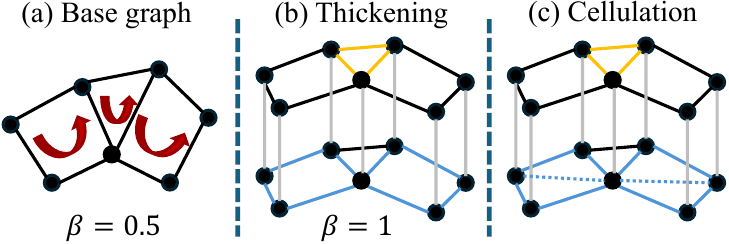}
    \vspace{-10pt}
    \caption{Graph Construction}
    \label{fig:GraphConstruction}
\end{figure}

The graph constructed in this manner, together with the measurement protocol introduced in Definition~10 of~\cite{he2025extractors}, enables the measurement of logical operators of a code.

\textbf{Bridge/Adapter.} To enable logical operator measurements across codes, we introduce the notion of bridge/adapter. A bridge typically refers to a system that connects the measurement graphs of two codes within the same family, whereas an adapter connects measurement graphs between different code families.

\begin{definition}[Bridge/Adapter~\cite{swaroop2024universal}]
Let $G_1 = (V_1, E_1)$ and $G_2 = (V_2, E_2)$ be two graphs. A bridge/adapter between $G_1$ and $G_2$ is defined as a set of pairwise non-overlapping edges $B$ that connect vertices in $V_1$ to vertices in $V_2$.
\end{definition}

Suppose we aim to measure the joint logical operator $L_1 L_2$, where $L_1$ and $L_2$ correspond to codes $\mathcal{Q}_1$ and $\mathcal{Q}_2$, respectively. The measurement procedure proceeds as follows. We first construct the measurement graphs $G_1$ and $G_2$ independently using Algorithm~\ref{alg:AncillaGraphConstruction}. We then set $d$ to be the minimum code distance of the two codes, and connect exactly $d$ vertices between the corresponding port sets $P_1$ and $P_2$ via a bridge/adapter.
The additional cycle checks introduced by the bridge/adapter are included as stabilizers in the deformed code. As a result, the joint logical operator $L_1 L_2$ is promoted to a stabilizer of the deformed code.

\section{Proofs} \label{Sec:Proofsapp}

\subsection{Proof Sketch of Theorem~\ref{thm:surgery-correctness}}
\label{proof:surgery-correctness}

\textbf{Theorem~\ref{thm:surgery-correctness}} (Correctness of surgery graphs).
Let $\mathfrak{G}$ be a feasible surgery graph for $(\mathcal{Q},L)$ under degree budgets $(d_q,d_c)$, and let $\widetilde{\mathcal{Q}}=\llbracket \mathfrak{G}\rrbracket_{\mathcal{Q},L}$. Then $\widetilde{\mathcal{Q}}$ is a well-defined stabilizer code in which $L$ is measurable as a stabilizer. Moreover, the graph expansion condition gives the distance-preservation guarantee of graph-based code surgery. Each ancilla qubit participates in at most $d_q$ cycle-basis checks and $O(1)$ other checks, and each generated vertex or cycle check has weight at most $d_c$.

\begin{proof}
The denotation $\llbracket \mathfrak{G}\rrbracket_{\mathcal{Q},L}$ follows the graph-based code-surgery construction described in Section~\ref{Sec:AncillaSystemasaGraphapp}. For each vertex $v\in V$, the construction adds a vertex check on the incident ancilla-edge qubits; for a port vertex $f(q)$, this check also contains the Pauli $L(q)$ on the corresponding original qubit. Since each ancilla edge is incident to exactly two vertices, the product of all vertex checks cancels the ancilla Pauli operators and leaves the logical operator $L$ on the original code block. Thus $L$ is promoted to a measurable stabilizer.

It remains to check that the generated stabilizers commute and preserve the intended distance and LDPC properties. For each original stabilizer check $S$, the dictionary entry $\mathcal{M}(S)$ connects the even set of ports where $S$ locally anti-commutes with $L$; multiplying $S$ by the corresponding ancilla-edge operators cancels exactly those anti-commutations with the vertex checks. The checks in $\mathcal{C}$ commute with vertex checks because every represented cycle intersects each vertex in an even number of incident edges, and their span fixes the gauge degrees of freedom introduced by the graph. These are precisely the stabilizer-commutation conditions of the graph-based surgery construction~\cite{williamson2024low,he2025extractors}.

The relative-expansion condition $\beta_d(G,P)\ge 1$ is the distance-preservation condition used by the same construction: any attempt to shorten a dressed logical operator through the ancilla graph must cross enough ancilla edges to offset the support removed from the original logical operator. Finally, the cellulation condition bounds the weight of every new cycle check by $d_c$, and the cycle-load condition contributes at most $d_q$ variable cycle-basis checks to an ancilla qubit. Its two vertex-check incidences and the contributions from path-matching repair and rectangular thickening checks are bounded independently of graph size. Cellulation preserves the tracked load on existing edges and gives each new local edge bounded load. Hence every ancilla qubit has degree $d_q+O(1)$, which preserves the QLDPC property when $d_q=O(1)$. Therefore $\widetilde{\mathcal{Q}}$ satisfies the claimed measurement, distance, and degree guarantees.
\end{proof}

\subsection{Proof of Theorem~\ref{th:ExpansionviaReduction}}
\label{proof:ExpansionviaReduction}

\textbf{Theorem~\ref{th:ExpansionviaReduction}} (Expansion via Reduction to $\mathcal{I}(n,r)$).
Let $n$ be even and fix $r\ge 3$. Let $G^{(r)}$ be the graph obtained with early stopping disabled, after completing Phase 1 of Algorithm~\ref{alg:ConditionedExpanderConstruction} and adding $r$ complete, independent Phase-2 matching layers. As $n\to\infty$, $G^{(r)}$ is an expander with probability $1-o(1)$.

\begin{proof}
With early stopping disabled, the $r$ complete matching layers added in Phase~2 are exactly a realization of $\mathcal{I}(n,r)$, the union of $r$ independent random perfect matchings on the common vertex set. The graph $G^{(r)}$ also contains the path-matching graph and all edges added during Phase~1. Hence $\mathcal{I}(n,r)$ is a spanning submultigraph of $G^{(r)}$.

For fixed $r\ge 3$, Corollary~1.4 of~\cite{friedman2003proof} gives that $\mathcal{I}(n,r)$ is an expander with probability $1-o(1)$ as $n\to\infty$. Adding edge instances cannot decrease any cut boundary, so the Cheeger constant of $G^{(r)}$ is at least that of $\mathcal{I}(n,r)$. Therefore $G^{(r)}$ is also an expander with probability $1-o(1)$.
\end{proof}

\subsection{Proof of Theorem~\ref{th:LowCongestionunderShortPaths}}
\label{proof:LowCongestionunderShortPaths}

\textbf{Theorem~\ref{th:LowCongestionunderShortPaths}} (Low Congestion under Short Paths).
Consider a sequence of calls to Algorithm~\ref{alg:cycle-tracking} on an $n$-vertex graph with at most $m$ edge instances, without externally replacing the maintained spanning forest. If every tree path used to generate a cycle has length $O(\log n)$, then the cycles added by these calls have congestion $O(\log n\cdot\log m)$ with probability at least $1/2$.

\begin{proof}
Fix an edge instance $e$. If $e$ is introduced during the sequence, it occurs once as the new edge of its fundamental cycle and then enters the maintained forest. In every later generated cycle containing $e$, the edge lies on the tree path and is removed with probability at least $\Omega(1/\log n)$, because that path has length $O(\log n)$ and Algorithm~\ref{alg:cycle-tracking} chooses the removed edge uniformly from it. The same removal bound applies immediately to an edge that was already in the initial forest. Once $e$ is removed from the forest, it cannot occur in a later fundamental cycle during the analyzed sequence: such a cycle contains only its newly inserted edge and edges of the current forest. Therefore, the probability that $e$ occurs in at least $w$ added cycles is at most
\[
\left(1 - \Omega\!\left(1/\log n\right)\right)^{w-1}.
\]
For $w=c\log n\cdot\log(2m)$ and a sufficiently large constant $c$, this probability is at most $1/(2m)$.

Applying a union bound over the at most $m$ edge instances, with probability at least $1/2$ no edge occurs in more than $O(\log n\cdot\log m)$ cycles added during the sequence. An external forest rebuild can reintroduce a removed edge and therefore begins a new epoch; the statement intentionally excludes such a rebuild.
\end{proof}

\subsection{Proof of Theorem~\ref{Sec:IndependenceoftheConstructedCycleBasis}}
\label{proof:IndependenceoftheConstructedCycleBasis}

\textbf{Theorem~\ref{Sec:IndependenceoftheConstructedCycleBasis}} (Independence of the Constructed Cycle Basis).
Let $r_1, \ldots, r_k$ denote the cycle basis obtained by applying the Decongestion Lemma to the path-matching graph, and let $r_{k+1}, \ldots, r_l$ be the additional cycles generated by Algorithm~\ref{alg:cycle-tracking} during the expander construction process. Then the set $\{r_1, \ldots, r_l\}$ is linearly independent.

\begin{proof}
Consider a linear combination
\[
a_1 r_1 + \cdots + a_k r_k + a_{k+1} r_{k+1} + \cdots + a_l r_l = 0.
\]
We show that all coefficients must be zero.

For each $i > k$, the cycle $r_i$ is formed using a newly added edge that does not appear in any previously constructed cycles. In particular, the most recently added cycle $r_l$ contains an edge that appears uniquely in $r_l$, implying that $a_l = 0$.

Applying the same argument inductively, we obtain $a_{l-1} = \cdots = a_{k+1} = 0$. Since $r_1, \ldots, r_k$ form a basis by construction, it follows that $a_1 = \cdots = a_k = 0$. This completes the proof.
\end{proof}

\subsection{Proof of Lemma~\ref{lemma:incre}}
\label{proof:incre}

\textbf{Lemma~\ref{lemma:incre}} 
The incremental partition procedure (Algorithm~\ref{alg:incremental-partition}) produces the same partition as the Greedy Partition algorithm (see Corollary~15 in~\cite{he2025extractors}) applied to the full set of cycles.

\begin{proof}
We prove the claim by induction on the number of cycles.
Let the cycles be ordered as $\mathcal{R} = \{r_l, \ldots, r_1\}$.

\textbf{Base case.} For $l = 1$, both algorithms place $r_1$ into the first group, and hence the resulting partitions are identical.

\textbf{Inductive step.} Assume processing $t-1$ cycles, both algorithms produce the same partition $\{P_1, \ldots, P_k\}$. Now we consider processing $t$ cycles.

In the Greedy Partition algorithm, observe that $r_t$ is always the last cycle to be checked within each partition group, and therefore does not affect the placement of any preceding cycles. Therefore, its placement can be viewed as occurring after all operations on the first $t-1$ cycles have been completed. By the induction hypothesis, the partition formed by the first $t-1$ cycles is identical to that produced by our incremental algorithm. Hence, placing $r_t$ after this step yields exactly the same outcome as in our algorithm.

Thus, after inserting $r_t$, the partitions produced by both algorithms remain the same. By induction, the two algorithms produce identical partitions for all cycles.
\end{proof}

\section{Degree-Aware Zigzag Cellulation}
\label{app:degree-aware-cellulation}

We give the exact cellulation used in Section~\ref{Sec:DegreeAwareCellulation}. Let
$C=(v_1,v_2,\ldots,v_\ell)$ be a cycle and let $d_c\geq 3$. If
$\ell\leq d_c$, the construction leaves $C$ unchanged. Otherwise, maintain
indices $a,b$ such that the unprocessed region is bounded by the path
$(v_a,v_{a+1},\ldots,v_b)$ and a closing edge instance $\{v_b,v_a\}$.
Initially, $a=1$, $b=\ell$, and the closing edge is the corresponding edge
of $C$.

While $b-a+1>d_c$, alternately cut a cell from the left and right. A left
cut sets $c=a+d_c-2$, adds a new edge instance $\{v_c,v_b\}$, emits the
cell
\[
(v_a,v_{a+1},\ldots,v_c,v_b),
\]
and sets $a\leftarrow c$. A right cut sets $c=b-d_c+2$, adds a new edge
instance $\{v_a,v_c\}$, emits the cell
\[
(v_a,v_c,v_{c+1},\ldots,v_b),
\]
and sets $b\leftarrow c$. In either case, the added edge becomes the closing
edge of the remaining region. When at most $d_c$ vertices remain, their
boundary is emitted as the final cell. New edges are edge instances, so this
construction remains valid when one is parallel to an existing edge.

\begin{proposition}[Zigzag cellulation properties]
\label{prop:zigzag-cellulation}
For a cycle of length $\ell>d_c$, the construction emits
\[
r=\left\lceil\frac{\ell-2}{d_c-2}\right\rceil
\]
cells and adds $r-1$ edge instances. Every cell has length at most $d_c$,
and their $\mathbb{F}_2$-sum is $C$. If $C$ belongs to a cycle basis
$\mathcal{B}$, replacing $C$ by these cells gives a cycle basis of the
augmented graph. Each original edge of $C$ occurs in one replacement cell,
each added edge occurs in two cells, and each vertex is incident to at most
two added edges.
\end{proposition}

\begin{proof}
Each nonfinal cut emits a $d_c$-edge cell and reduces the number of vertices
on the remaining boundary by $d_c-2$. This gives the stated values of $r$
and $r-1$, and the final boundary has between three and $d_c$ edges. Every
added edge bounds the cell that introduces it and the next cell in the
zigzag, while every original edge of $C$ bounds exactly one emitted cell.
The added edges therefore cancel in the sum of all cells, leaving exactly
$C$.

It remains to check the basis claim. The added edges form a chain between
successive cells. In any linear combination of the cells that contains no
added edge, the coefficient of each cell must equal that of its neighbor.
Thus either no cell is selected or all cells are selected, in which case
their sum is $C$. Any dependence after replacing $C$ in $\mathcal{B}$ would
therefore induce a dependence in the original basis. The replacement set is
independent and has the correct size because adding $r-1$ edges increases
the cycle-space dimension by $r-1$.

Alternating the side of each cut makes a chord endpoint persist for at most
the next cut, so no vertex receives more than two added incidences. The
cellulation consequently adds $r-1$ ancilla qubits and replaces one cycle
check by $r$ checks, for a net increase of $r-1$ checks. It does not increase
the tracked cycle load on an original edge, and every added edge has
cycle-check load two.
\end{proof}

\section{Integrated Synthesis Driver}
\label{sec:CongestionAwareExpanderConstruction}

Algorithm~\ref{alg:CongestionAwareExpanderConstruction} summarizes how the
components of Sections~\ref{Sec:OptimizingExpanderConstructionviaConditionedRandomization}--\ref{Sec:IncorporatingDegreeConstraintsintoPracticalCodeDesign}
operate on shared state. The conditioned proposal schedule below means the
edge instances generated, in order, by the two sampling phases of
Algorithm~\ref{alg:ConditionedExpanderConstruction}; the coupled stopping
rule replaces that algorithm's fixed-$\beta$ stopping rule.

\begin{algorithm}[h]
\caption{Integrated Congestion- and Degree-Aware Synthesis}
\label{alg:CongestionAwareExpanderConstruction}

\KwIn{Path-matching graph $G_0$ with $(P,f,\mathcal{M})$, maximum matching layers $\tau$, degree budgets $(d_q,d_c)$}
\KwOut{Final surgery graph IR $\mathfrak{G}_{\mathrm{final}}$}

Compute a low-congestion cycle basis $\mathcal{R}$ of $G_0$ and a spanning forest $T$\;
Set $\mathfrak{G}\leftarrow(G_0,P,f,\mathcal{M},\mathcal{R})$\;
Partition $\mathcal{R}$ under budget $d_q$ by repeated application of Algorithm~\ref{alg:load-incremental-partition}, obtaining partition $\Pi$, load trackers $\Lambda$, and layer count $t$\;
$\ell_{\mathrm{exp}}\leftarrow\lceil 2/\lambda_2(G_0)\rceil$; $\ell_{\mathrm{cong}}\leftarrow t$\;

\ForEach{edge instance $e$ from the conditioned proposal schedule, up to $\tau$ matching layers}{
    Set $(\mathfrak{G},T)\leftarrow\mathsf{expand}_{e}(\mathfrak{G},T)$ using Algorithm~\ref{alg:cycle-tracking}, and let $r$ be the new cycle, if any\;
    \If{$r\neq\emptyset$}{
        Update $(\Pi,\Lambda,t)$ with Algorithm~\ref{alg:load-incremental-partition}$(\Pi,\Lambda,r,d_q)$\;
    }
    $\ell_{\mathrm{exp}}\leftarrow\lceil 2/\lambda_2(G)\rceil$; $\ell_{\mathrm{cong}}\leftarrow t$\;
    \If{$\ell_{\mathrm{cong}}\geq\ell_{\mathrm{exp}}$}{
        \textbf{break}\;
    }
}

$\ell\leftarrow\max(\ell_{\mathrm{exp}},\ell_{\mathrm{cong}})$\;
$\mathfrak{G}\leftarrow\mathsf{thicken}_{\ell}(\mathfrak{G})$\;
$\mathfrak{G}\leftarrow\mathsf{decongest}_{\Pi}(\mathfrak{G})$, padding $\Pi$ with empty parts when $\ell>t$\;
$\mathfrak{G}_{\mathrm{final}}\leftarrow\mathsf{cellulate}_{d_c}(\mathfrak{G})$ using Appendix~\ref{app:degree-aware-cellulation}\;
\Return $\mathfrak{G}_{\mathrm{final}}$\;

\end{algorithm}

The loop may stop at the expansion--congestion balance point or after the
proposal budget is exhausted. Exhausting the budget is not a failure: the
achieved base expansion determines $\ell_{\mathrm{exp}}$, and thickening by
$\ell$ supplies the final relative-expansion condition. The driver omits the
internal sampling, forest-rebuild, and cell-placement details already given
by the component procedures.